\documentclass[]{hdsr}

\graphicspath{{figs/}}

\usepackage{amsmath}
\usepackage{bm}
\usepackage{algorithm}
\usepackage{amssymb}
\usepackage{algpseudocode}
\newtheorem{definition}{Definition}
\newtheorem{theorem}{Theorem}[section]
\newtheorem{proposition}{Proposition}[section]
\newtheorem{remark}{Remark}[section]
\newtheorem{corollary}{Corollary}[section]

\newtheorem{lemma}[theorem]{Lemma}

\usepackage{tikz}
\usepackage{hyperref}
\hypersetup{
    colorlinks=true,
    linkcolor=blue,
    filecolor=magenta,      
    urlcolor=cyan,
    pdftitle={Overleaf Example},
    pdfpagemode=FullScreen,
    citecolor = blue,
    }
\usetikzlibrary{positioning}
\usetikzlibrary{automata,positioning,fit,shapes.geometric,backgrounds}
\usetikzlibrary{shapes,calc,arrows}
\tikzstyle{process} = [rectangle, rounded corners, minimum width=2.5cm, minimum height=1cm,text width = 2.4cm, text centered, draw=black, fill=black!9]
\tikzstyle{box} = [rectangle, rounded corners, minimum width=2.5cm, minimum height=1cm,text width = 2.4cm, text centered, draw=black, fill=black!4]

\tikzstyle{io} = [rectangle, minimum width=2cm, minimum height=1cm, text centered, text width = 2 cm, draw=black, fill=black!4]

\tikzstyle{arrow} = [thick,->,>=stealth]

\begin{document}

\newgeometry{bottom=1.5in}


\begin{center}

  \title{Differentially Private Linear Regression with Linked Data}
  \maketitle

  \thispagestyle{empty}
  
  \vspace*{.2in}

  \begin{tabular}{cc}
    Shurong Lin\upstairs{\affilone,*}, Elliot Paquette\upstairs{\affiltwo}, Eric D. Kolaczyk\upstairs{\affiltwo}
   \\[0.25ex]
   {\small \upstairs{\affilone} Department of Mathematics and Statistics, Boston University} \\
   {\small \upstairs{\affiltwo} Department of Mathematics and Statistics, McGill University} 
  \end{tabular}

  \emails{
    \upstairs{*} shrlin@bu.edu. The research was supported in part by the U.S. Census Bureau Cooperative Agreement CB20ADR0160001 and Canadian NSERC RGPIN-2023-03566. The authors would like to thank Adam Smith (Boston University) for all the helpful discussions and comments.
    }
  \vspace*{0.4in}

\begin{abstract}
There has been increasing demand for establishing privacy-preserving methodologies for modern statistics and machine learning. 
Differential privacy, a mathematical notion from computer science, is a rising tool offering robust privacy guarantees. 
Recent work focuses primarily on developing differentially private versions of individual statistical and machine learning tasks, with nontrivial upstream pre-processing typically not incorporated. 
An important example is when record linkage is done prior to downstream modeling. 
Record linkage refers to the statistical task of linking two or more datasets of the same group of entities without a unique identifier. 
This probabilistic procedure brings additional uncertainty to the subsequent task. In this paper, we present two differentially private algorithms for linear regression with linked data. 
In particular, we propose a noisy gradient method and a sufficient statistics perturbation approach for the estimation of regression coefficients.
We investigate the privacy-accuracy tradeoff by providing finite-sample error bounds for the estimators, which allows us to understand the relative contributions of linkage error, estimation error, and the cost of privacy. 
The variances of the estimators are also discussed.
We demonstrate the performance of the proposed algorithms through simulations and an application to synthetic data.
\end{abstract}
\end{center}

\vspace*{0.15in}
\hspace{10pt}
  
\copyrightnotice

\section*{Media Summary} 
Differential privacy is a mathematical framework for ensuring the privacy of individuals in datasets. 
{\color{black} 
It mitigates the privacy risk of disclosing sensitive information about individuals within the dataset during data analysis.}
Under such a framework, we are interested in finding the relationship between two variables (via statistical regression) \textit{after} they are linked from two data sources with uncertainties.
A pre-processing procedure of linking datasets is called record linkage, and the uncertainties should be taken into account in the downstream analysis.
In the article, we propose two algorithms that satisfy differential privacy for regression estimation problems with linked data. 
The theoretical results regarding privacy guarantees and statistical accuracy are provided. We demonstrate the performance of the proposed algorithms through simulations and an application.

\section{Introduction}
\label{sec:intro}
Data for the same group of entities are often scattered across different resources, lacking unique identifiers for perfect linkage. 
To conduct statistical modeling or inference on the integrated information, it is necessary to probabilistically link multiple datasets by comparing the common quasi-identifiers (e.g., names, gender, address) as a pre-processing step.  
Such a procedure
is called record linkage (RL), also known as entity resolution, or data matching \citep{Christen2012data}, which is an essential component of data integration in big data analytics \citep{Dong2015}.
Thanks to its wide application in many disciplines such as public health and official statistics,
record linkage has been studied for decades.
Earlier pioneering works include
\citet{Newcombe1959, fellegi1969,Jaro1989AdvancesIR}.
In addition, record linkage is frequently used in current practice.
{\color{black} 
The U.S. Census Bureau has a
long tradition using record linkage methodology for
multiple endeavors.
A current prominent example is the Decennial Census \citep{census2022rl}. 
In this context, record linkage involves using administrative records and other data sources to improve data quality, with efforts underway to construct a comprehensive ``reference database'' including individuals from multiple administrative records. 
}
A recent review paper,
\citet{Olivier2022RL}, provided a comprehensive summary of record linkage. 
Broadly speaking, there are two perspectives regarding record linkage \citep{Chambers2019SmallAE}: (1) the primary viewpoint concerns how to link the records; (2) the secondary perspective is focused on how to propagate the uncertainty to the downstream statistical learning tasks \textit{after} the linkage has been determined. Our focus in this paper will adopt the second of these two perspectives.  

Closely related to record linkage is data privacy. 
In the area of privacy-preserving record linkage (PPRL), two or more private datasets owned by different organizations are linked without revealing the data to one another \citep{Hall2010pprl, Christen2020}.
{\color{black} The outcome of PPRL is the information regarding which pairs or sets are matched.}
PPRL, in turn, is associated with secure multiparty computation (SMPC) in that SMPC techniques are commonly used to solve PPRL problems \citep{Kuzu2013, He2017pprl, Rao2019}. 
PPRL to date only engages in the private linkage process from a primary perspective, without concerning how the linkage uncertainties would impact the downstream analysis. 
On the contrary, the secondary perspective is to modify the statistical tools to account for the linkage uncertainty. Our goal is to incorporate privacy into the secondary perspective of record linkage, which is different from yet complementary to PPRL or SMPC.

Privacy concerns have, if anything, become significantly more exacerbated with the emergence of individual-level big data. Releasing information about a sensitive dataset is subject to a variety of privacy attacks \citep{Dwork2017Exposed}. Therefore, there has been a growing demand for establishing robust privacy-preserving methodologies for modern statistics and machine learning. A mathematical framework proposed by \citet{Dwork2006}, differential privacy (DP), is now considered the gold standard for rigorous privacy protection and has made its way to broad application in industry \citep{apple2017dp, ms2020dp, google2021dp} and the public sector \citep{census2021dp}.
The literature on differential privacy has been flourishing in recent years and the interface of differential privacy and statistics has started to draw increasing attention from the statistics community. 

Recent work on differential privacy focuses primarily on individual statistical and machine learning tasks, with nontrivial upstream pre-processing, such as record linkage, typically not incorporated. 
In this paper, we consider the linear regression problem, i.e.,
\begin{equation}
\label{model:lm}
    \bm y = X\bm \beta + \bm e, \quad \bm e \sim \mathcal{N}(0, \sigma^2 I_n)
\end{equation}
but where $X$ and $\bm y$ are observed in two separate datasets.  As a result, rather than having $X$ and $\bm y$ in hand, we are instead provided with a pair $X$ and $\bm z$.  Here $\bm z$ is a permutation of $\bm y$ resulting from record linkage performed by an external entity, who also supplies a minimum amount of information about the linkage accuracy.  
In the regression procedure, we take into account the linkage uncertainty as well as offer differential privacy guarantees. 
As shown in Figure \ref{Flowchart} which depicts the pipeline of the problem we consider, we assume that an external analyst conducts record linkage a priori.
From there, we aim to devise a private estimator for the regression coefficients of ultimate interest with the help of differential privacy.

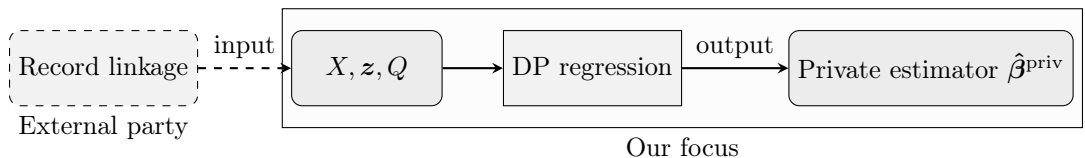
\begin{figure}[H]
\centering
\begin{tikzpicture}
    \node (rl) at (-3.5,1)  [rectangle,dashed, rounded corners, minimum width=2.5cm, minimum height=1cm, text centered, draw=black, label=below:External party, fill=gray!10] {Record linkage};
    \node (input) at (0,1)  [rectangle, rounded corners, minimum width=2cm, minimum height=1cm, text centered, draw=black,  fill=gray!15] {$ X, \bm z, Q$};
    \node (process) at (3,1)  [rectangle, minimum width=2cm, minimum height=1cm, text centered, draw=black,  fill=gray!10] {DP regression};
    \node (output) at (7.5,1)  [rectangle,  rounded corners, minimum width=2cm, minimum height=1cm, text centered, draw=black,  fill=gray!15] {Private estimator ${\bm{\hat\beta}^{\text{priv}}}$};
    \draw [arrow] (input) -- node [midway,above] {} (process);
    \draw [arrow] (process) -- node [midway,above] {output} (output);
    \draw [dashed, arrow] (rl) -- node [midway,above] {input} (input);
    
    \begin{scope}[on background layer]
    \node [fit=(input) (output), rectangle, fill=gray!5, draw=black, fill opacity=0.5, label=below:Our focus, minimum height=45, minimum width = 290] {};
    \end{scope}
    
\end{tikzpicture}
\caption{{\color{black} Pipeline of private regression with linked data.}}
\label{Flowchart}
\end{figure}

Specifically, we propose two algorithms for linear regression after record linkage to meet differential privacy: (1) post-RL noisy gradient descent (NGD), and (2) post-RL sufficient statistics perturbation (SSP). 
Our work builds on the seminal work by \citet{lahiri2005} where an estimator is proposed for linear regression with linked data in a non-privacy-aware setting. We construct a private estimator, ${\bm{\hat\beta}^{\text{priv}}}$, by deploying differential privacy tools to achieve privacy protections.
To the best of our knowledge, our work is the first one in the literature to consider a statistical model after record linkage in a privacy-aware setting. 

The two proposed algorithms also extend the noisy gradient method \citep{Bassily2014PrivateER} and 
{\color{black} the ``Analyze Gauss'' algorithm \citep{DworkTT014}}, which are applied to linear regression, to additionally handle the presence of linkage errors.
Prior works \citep{Sheffet17, wang2018, Bernstein2019,Cai2019TheCO, Alabi2022} on differentially private linear regression do not consider possible record linkage pre-processing. 
If the data are linked beforehand, directly applying their algorithms to the imperfectly linked data is not ideal.
It is well known that overlooking the linkage errors leads to substantial bias even with a high linkage accuracy \citep{Neter1965TheEO, Scheuren1993}. 
{\color{black} Figure \ref{fig:rlexample} 
showcases a toy example of record linkage, where mismatches, if treated as true, change the sign of the slope estimate.
}
Our illustrative application later in the paper confirms this, where around 90\% of the records are correctly linked, and the estimators ignoring linkage errors end up with large biases.
\begin{figure}[H]
\begin{tikzpicture}
\centering
  \node[draw,rectangle,label=below:Dataset A,] at (0,0) (leftTable) {
    \begin{tabular}{cc|c}
      \hline
      first name & last name & $x$ \\
      \hline
      Shurong & Lin & 1 \\
      \hline
      Erik & K & 2 \\
      \hline
      Elliot & P & 3 \\
      \hline
      S & Li & 4 \\
      \hline
    \end{tabular}
  };
  
  \node[draw,rectangle,label=below:Dataset B] at (7,0) (rightTable) {
    \begin{tabular}{cc|c}
      \hline
      first name & last name & $y$ \\
      \hline
      S & L & 2 \\
      \hline
      Eric & K & 4 \\
      \hline 
      Eliot & P & 6 \\
      \hline
      Sharon & Li & 8 \\
      \hline
    \end{tabular}
  };

  \draw[arrow,dashed] (2.5,0.5) -- (4.5,-1);
  \draw[arrow,dashed] (2.5, -1) -- (4.5, 0.5);
  \draw[arrow] (2.5, 0) -- (4.5, 0);
  \draw[arrow] (2.5, -0.5) -- (4.5, -0.5);
\end{tikzpicture}
\caption{{\color{black} A toy example of record linkage with mismatches (dashed links). \\
The true dataset $(X, \bm y)$  is $\{(1,2),(2,4),(3,6),(4,8)\}$, yielding a slop estimate $\hat\beta_1=2$, while the linked set $(X, \bm z)$ is given by $\{(1,8),(2,4),(3,6),(4,2)\}$, yielding  $\hat\beta_1=-1.6$.}}
\label{fig:rlexample}
\end{figure}

{\color{black} 
Accompanying the estimators resulting from our algorithms, we provide mean-squared error bounds under typical regularity assumptions and record linkage schemes. 
When no linkage errors are present (i.e., a special case in our scenario), our result in Theorem \ref{thm:ngd-bound} improves upon the noisy gradient method proposed in \citet{Cai2019TheCO} by using zero-concentrated differential privacy (zCDP, \citet{Bun2016ConcentratedDP}) to enable tighter bounds on privacy cost (see Lemma \ref{lem:tighter-comp}). 
Additionally, we have presented (approximate) theoretical variances for ${\bm{\hat\beta}^{\text{priv}}}$ resulting from both proposed algorithms. There appear to be very few other works that have addressed the issue of uncertainty. 
Two that we are aware of are \citet{Alabi}, who provided confidence bounds for the univariate case, and \citet{Sheffet17}, who provided confidence intervals dependent on differential privacy noise. Our work focuses on the multivariate case and appears to be the first to directly work on exact variances rather than relying on bounds.
}

The remainder of this paper is organized as follows. Section \ref{sec:Preliminaries} provides preliminaries on linear regression with linked data and differential privacy. 
We propose our two algorithms in Section \ref{sec:DP-algo} and present the relevant theoretical results 
in Section \ref{sec:theoreical}. 
In Section \ref{sec:numerical}, we conduct a series of simulation studies and an application to synthetic data. Section \ref{sec:discussion} concludes and discusses future work. Complete proofs of all theorems can be found in the supplementary materials.


\section{Preliminaries}
\label{sec:Preliminaries}
In this section, we review the background results of linear regression after record linkage upon which we build our work, and fundamental concepts from differential privacy.
Related work on linear regression with linked data and record linkage with privacy awareness are discussed.

\subsection{Linear Regression with Record Linkage}
\label{subsec:lrwrl}
Let $(X, \Phi_{X}) $ and $ (\bm y, \Phi_{\bm y}) $ be two datasets that refer to the same group of $n$ entities, with unknown one-to-one correspondence. The quasi-identifiers $\Phi_{X}$ and $\Phi_{\bm y}$ are used to perform the linkage procedure. Let $(X, \bm z)$ be the linked data where $\bm z$ is a permutation of $\bm y$. Consider the following model for $\bm z$:
\begin{equation}
\label{model:z}
    \mathbb{P}(z_i = y_j) = q_{ij},\quad i,j = 1, \dots, n,
\end{equation}
then $\sum_{j=1}^n q_{ij} = 1$ for all $i$ and $ \sum_{i=1}^n q_{ij} = 1$ for all $j$.
Thus, $q_{ii}$ is the probability of the $i$th record being linked correctly.
Let $Q = (q_{ij})$, which we call the matching probability matrix (MPM), a doubly stochastic matrix. The matrix $Q$ can be estimated, for example, through bootstrapping \citep{Chipperfield2015, Chipperfield2020}.
In some cases, estimating $Q$ can require inference on only a single parameter (e.g., in the exchangeable linkage error (ELE) model described in Section \ref{subsubsec:rl-schemes}).

For the fixed-design homoskedastic linear model (\ref{model:lm}), when inference is done after record linkage based on $(X,{\bm z})$,
\citet{lahiri2005} proposed an unbiased estimator
\begin{equation}
    {\hat {\bm \beta}^{\text{RL}}} = (W^\top  W)^{-1} W^\top  \bm z,
\label{est: rlest}
\end{equation}
where $W = QX$.
{\color{black} 
Let $\bm w_i$ be the $i$-th row vector of $W$, then $\bm w_i = \sum_{j=1}^nq_{ij}\bm x_j$. 
Note that $\mathbb{E}(z_i) = \bm{w}_i^\top \bm{\beta}$, where the expectation is taken over both linkage uncertainties and $\bm y$.
Transforming $X$ into $W$ offers bias correction for regression estimation after record linkage.
}

In addition, the variance of ${\hat {\bm \beta}^{\text{RL}}}$ is given by
\begin{equation}
    \Sigma^{\text{RL}} \stackrel{def}{=} \operatorname{Var}({\hat {\bm \beta}^{\text{RL}}}) = (W^\top  W)^{-1} W^\top  \Sigma_{\bm z} W (W^\top  W)^{-1},
    \label{eq:cov-rl}
\end{equation}
where $\Sigma_{\bm z} \stackrel{def}{=} \operatorname{Var}(\bm z)$.  \citet{lahiri2005} provide the following characterization of the first two moments of $\bm z$.
\begin{lemma}[Theorem A.1, \cite{lahiri2005}]
\label{lem:cov-z}
Under the model described by (\ref{model:lm}) and (\ref{model:z}), we have for $i,j = 1, \dots, n$
\begin{itemize}
    \item $\mathbb{E}(z_i) = \bm w_i^\top \bm\beta$;
    \item $
    \operatorname{Var}(z_i) = \sigma^2 + \bm\beta^\top  A_i \bm\beta \mbox { with }
    A_i = \sum_{j=1}^n q_{ij}(\bm x_j - \bm w_i) (\bm x_j - \bm w_i)^\top
    $; 
    \item $
    \operatorname{Cov}(z_i, z_j) = \bm\beta^\top  A_{ij} \bm\beta \mbox { with }
    A_{ij} = \sum_{u=1}^n \sum_{v \neq u}^n q_{iu}q_{jv}(\bm x_i - \bm w_u)(\bm x_j - \bm w_v)^\top
$.
\end{itemize}
\end{lemma}

Note that $\Sigma_{\bm z}$ involves the true coefficients $\bm\beta$ and $\Sigma_{\bm z} = \sigma^2 I_d + h(\bm\beta, Q, X)$ where $h(\bm\beta, Q, X)$ is a function of $\bm\beta, Q, X$ as elaborated in Lemma~\ref{lem:cov-z}. Compared to the covariance of $\bm y$, $\Sigma_{\bm z}$ has an additional component $h(\bm\beta, Q, X)$ due to the uncertainty of record linkage.

\subsubsection{Structural Schemes of MPM}
\label{subsubsec:rl-schemes}
The matching probability matrix (MPM) $Q$ is generally assumed to have a simple structure. Two schemes used commonly in the literature are as follows.
\paragraph{\textbf{Blocking Scheme}} 
It is assumed that the MPM is a block diagonal matrix, which means the true matches only happen within blocks. Blocking significantly reduces the number of pairs for comparison and allows scalable record linkage. This scheme is used in almost all real-world applications, and different methods for blocking have been developed \citep{Christen2012data, Steorts2014a, Christophides2020an}.

\paragraph{\textbf{Exchangeable Linkage Errors (ELE) Model}}
The ELE model \citep{Chambers2009RegressionAO} assumes homogeneous linkage accuracy and errors:
    \begin{equation}
    \label{model:ele}
    \begin{aligned}
        & \mathbb{P}(\text{correct linkage}) = q_{ii} = \gamma,\\
        & \mathbb{P}(\text{incorrect linkage}) = q_{ij} = \frac{1-\gamma}{n-1} \text{ for } i \neq j.
    \end{aligned}       
    \end{equation}
The ELE model has been adopted in recent works, such as \citet{Chambers2019SmallAE,Chambers2022}, for various estimation problems.
Even though (\ref{model:ele}) may oversimplify the reality, it is a representative model for a secondary analyst who has minimum information about the linkage quality. 
When blocking is used, 
the homogeneous linkage
accuracy assumption is imposed within individual blocks. In other words, it still allows heterogeneous linkage accuracy between blocks.

\subsection{Differential Privacy}
Let $\mathcal{X}$ be some data space, and $D, D' \in \mathcal{X}^n$ be two neighboring datasets of size $n$ which only differ in one record. Such a relation is denoted by $D\sim D'$. 
\begin{definition}[$(\epsilon,\delta)$-DP, \cite{DworkR14}] For $\epsilon > 0, \delta \geq 0$, a randomized algorithm $A$: $\mathcal{X}^n \rightarrow  \mathcal{R}$ is  
$(\epsilon,\delta)$-differentially private if, for all  $D  \sim D' \in \mathcal{X}^n$ and any $\mathcal{O} \subseteq \mathcal{R}$,
\begin{equation}
\label{def:dp}
    \mathbb{P}(A(D )\in \mathcal{O}) \leq e^\epsilon \cdot \mathbb{P}(A(D')\in \mathcal{O}) + \delta.
\end{equation}
\end{definition}

The expression (\ref{def:dp}) controls the distance between the output distributions on two neighboring datasets through the privacy budget $\epsilon$ and $\delta$.
Intuitively, differential privacy ensures that $D$ is not distinguishable from $D'$ based on the outputs. Thus, $\epsilon$ should be small enough for the privacy level to be meaningful. Typically, $\epsilon \in (10^{-3},10)$ and $\delta = o(1/n)$.

Differential privacy enjoys the following properties that facilitate the construction of differentially private algorithms.
\begin{proposition}[Basic composition, \cite{DworkR14}]
\label{prop:comp}
    If $f_1$ is $(\epsilon_1,\delta_1)$-DP and $f_2$ is $(\epsilon_2,\delta_2)$-DP,
    then $f: = (f_1, f_2)$ is $(\epsilon_1 + \epsilon_2,\delta_1 +\delta_2)$-DP.
\end{proposition}
\begin{proposition}[Post-processing, \cite{DworkR14}]
    If $f$ is $(\epsilon,\delta)$-DP, for any deterministic mapping $g$ that takes $f(D)$ as an input, then $g(f(D))$ is $(\epsilon,\delta)$-DP.
\end{proposition}

Generally, a differentially private algorithm is constructed by adding random noise from a certain structured distribution, such as the Laplace or Gaussian distributions. 
A notion central to the amount of noise we add is the sensitivity of the estimation function we desire to release privately.
\begin{definition}[$\ell_2$-sensitivity]
Let $f: \mathcal{X}^n \rightarrow \mathbb{R}^d$ be an algorithm. The $\ell_2$-sensitivity of $f$ is defined as
\begin{equation}
\Delta_f = \max_{D\sim D' \in \mathcal{X}^n } \|f(D)-f(D')\|_2.
\end{equation}
\end{definition}
The sensitivity of a function characterizes how much the output would change if one record in the dataset changes. To achieve $(\epsilon, \delta)$-DP, the amount of noise we need depends on both the budget and the sensitivity.
The Gaussian Mechanism is a canonical example that will be employed herein, which does just that. 
\begin{lemma}[Gaussian mechanism, \citet{DworkR14}]
Let $0<\epsilon <1$ and $\delta > 0$. For an algorithm $f$ on the dataset $D$, the Gaussian Mechanism $A(\cdot)$ defined as
\begin{equation}
    A(D) := f(D) + u,
\end{equation}
where $u\sim \mathcal{N} (0, 2\ln(1.25/\delta) ( \Delta_f/ \epsilon)^2)$,
is $(\epsilon,\delta)$-DP.
\label{lem: gausmech}
\end{lemma}

{\color{black} 
Combining the basic composition rule and the Gaussian mechanism, for a sequence of functions $(f_1,f_2,\dots,f_T)$, let
\begin{equation*}
u_t \sim \mathcal{N}\left(0, \frac{2T^2\Delta_t^2\ln(1.25T/\delta)}{\epsilon^2}\right),
\end{equation*}
where $\Delta_t$ is the $\ell_2$-sensitivity of $f_t$. Then, $A := (f_1+u_1,f_2+u_2,\dots,f_t+u_T)$ satisfies $(\epsilon, \delta)$-DP. 
However, as $T$ increases, this construction tends to add more noise than necessary due to the loose composition. 
Instead, we could utilize zero-concentrated differential privacy (zCDP, \citet{Bun2016ConcentratedDP}), another variant of DP, to achieve tighter composition for $(\epsilon,\delta)$-DP. The following Lemma essentially captures the results from \citet{Bun2016ConcentratedDP}, formulated for our purposes.

\begin{lemma}[Better composition for $(\epsilon,\delta)$-DP via zCDP]
Let $\epsilon>0, \delta>0$.
For a sequence of functions $(f_1,f_2,\dots,f_T)$, let 
\begin{equation}
    u_t \sim \mathcal{N}\left(0, \frac{T\Delta_t^2}{2\rho}\right),
   \label{eq:tighter-comp-rho} 
\end{equation}
with $\rho:= \epsilon+2\ln(1/\delta)-2\sqrt{(\epsilon+\ln(1/\delta))\ln(1/\delta)}$.
Then, the randomized algorithm $A:=(f_1+u_1,f_2+u_2,\dots,f_T+u_T)$
satisfies $(\epsilon, \delta)$-DP. If $\epsilon \leq \frac{8\ln(1/\delta)}{2+\sqrt{2}}$, it suffices to have
\begin{equation}
    u_t \sim \mathcal{N}\left(0, \frac{4T\Delta_t^2\ln(1/\delta)}{\epsilon^2}\right).
\label{eq:tighter-comp}
\end{equation}
\label{lem:tighter-comp}
\end{lemma}

Please refer to the supplementary materials for details. Since, in most practical budget settings, we have $\epsilon \leq \frac{8\ln(1/\delta)}{2+\sqrt{2}}$, we will apply (\ref{eq:tighter-comp}) for composition and analysis in the rest of the paper, acknowledging that (\ref{eq:tighter-comp-rho}) is valid for all parameter ranges.

In Section \ref{sec:DP-algo}, we shall employ Lemmas \ref{lem: gausmech} and \ref{lem:tighter-comp} in devising two distinct algorithms for linear regression after record linkage.
}
\subsection{Related Work}
\label{subsec:related-work}
Linear regression with linked data is a fundamental statistical task that has been explored in various articles. \citet{Scheuren1993} initially considered the linkage model (\ref{model:z}) for linear regression  and proposed an estimator that is not generally unbiased. Later, \citet{lahiri2005} introduced an exactly unbiased OLS-like estimator given in (\ref{est: rlest}) with an expression for the variance, which outperformed the approach by \citet{Scheuren1993}. Besides, \citet{Chambers2009RegressionAO, Zhang2021linkage} offered a few other estimators. 
According to their simulation studies, some of the estimators provided performance that was at most similar, but not noticeably better, compared to the one proposed by \citet{lahiri2005}. 
Yet, \citet{Zhang2021linkage} relaxed the condition by not assuming that the probability of correct linkage, $q_{ii}$ in the model (\ref{model:z}), can be obtained or estimated. 
For more extensive reviews of this literature,
\citet{Wang2022reg} gave an account of the recent development of various methods on regression analysis with linked datasets. \citet{Chambers2022} reviewed current research on robust regression of linked data.

On the other hand, there is ongoing research on privacy-preserving record linkage (PPRL) in the field of computer science. PPRL aims to privately link multiple sensitive datasets held by different organizations when they are unwilling or not permitted to share their data with external parties due to privacy and confidentiality concerns. To achieve privacy protection, techniques such as SMPC and DP are combined with machine learning and deep learning methods for conducting PPRL \citep{ Christen2020, Divanis2021, Ranbaduge2022PPRL}. 
PPRL primarily concerns data leakage during the linkage process and produces a linked dataset that can be used for further analysis, yet most applications treat the linked data as if there were no linkage errors. Neither the uncertainty propagation nor private release of the downstream analysis is considered within the scope of PPRL. 
 
Note that there are several articles on privacy-preserving analysis on vertically partitioned databases. In these databases, the attributes are distributed among multiple parties, but common unique identifiers exist to facilitate data linkage across the different parties. Unlike probabilistic record linkage, vertically partitioned databases do not involve linkage errors.
\citet{Du2004, Sanil2004, Hall2011SecureML, Gasc2017} discussed the
implementations of privacy-preserving linear regression protocols that prevent data disclosure across organizations, whereas \citet{Dwork2004PPVRD} considered data mining from the perspective of the private release of statistical querying in a spirit similar to our work.

\section{Differentially Private Algorithms}
\label{sec:DP-algo}
The unbiased and simply structured estimator provided in (\ref{est: rlest}) with a known closed-form variance makes it a suitable prototype to construct our private estimators. 
We introduce two differentially private algorithms in the following, based on
(1) noisy gradient descent, and (2) sufficient statistics perturbation. As the names suggest, we mitigate privacy risk by perturbing either the gradient or sufficient statistics during the computation of the linear model.
Hereafter, if not specified otherwise, $\|\cdot \|$ denotes the 2-norm.

\subsection{Post-RL Noisy Gradient Descent}
Gradient descent methods are ubiquitous in scientific computing for numerous optimization problems. 
Within the framework of differential privacy,
\citet{Bassily2014PrivateER} provided a noisy variant of the classic gradient descent algorithm. It was later adapted by
\citet{Cai2019TheCO} to solve the classic linear regression problem with faster convergence. 
Leveraging the work by  \citet{Bassily2014PrivateER, Cai2019TheCO}, we tailor the noisy gradient method for the post-RL linear regression model for $(X,\bm z)$ based on (\ref{model:lm}) and (\ref{model:z}).

Let $\mathcal{L}_n(\bm \beta) \stackrel{def}{=} \frac{1}{2n}(\bm z - W\bm\beta)^\top (\bm z - W\bm\beta)$ be the loss function, where recall $W = QX$.  The minimizer of $\mathcal{ L}_n(\bm\beta)$ is the non-private 
 RL estimator proposed by \citet{lahiri2005}. Let  $\Pi_R(\bm r)$  denote the projection of $\bm r \in\mathbb{R}^s $ onto the $\ell_2$ ball $\{\bm r\in\mathbb{R}^s: \|\bm r\|\leq R\}$.  The post-RL noisy gradient descent (NGD) algorithm is defined as follows.

\begin{algorithm}[ht]
\caption{Post-RL Noisy Gradient Descent}
\begin{algorithmic}[1]
\Require Linked dataset $(X,\bm z)$ and matching probability matrix $Q$, 
privacy budget ($\epsilon, \delta$), noise scale factor $B$, step size $\eta$, number of iterations $T$, truncation level $R$, feasibility $C$, initial value $\bm{\beta}^0$.

\State Let $ W = QX$.
\For{$t=0$ to $T-1$}
    \State Generate $\bm u_t \,{\sim}\, \mathcal{N}\left(0, \omega^2 I_d \right)$ where $\displaystyle \omega = \frac{2\eta  B\sqrt{T\ln (1/\delta)}}{n\epsilon}$.
    \State Compute 
    \begin{equation}
    \label{eq:update}
       \bm{\beta}^{t+1} = \Pi_C(\bm{\beta}^t - 
    \frac{\eta}{n}\sum_{i=1}^n(\bm w_i^\top\bm{\beta}^t - \Pi_R(z_i))\bm w_i + \bm u_t).
    \end{equation}
\EndFor
\Ensure $\bm{\hat\beta}^{priv} = \bm{\beta}^T$. 

\end{algorithmic}
\label{alg:ngd}
\end{algorithm}

Algorithm \ref{alg:ngd} is a modified version of the projected gradient descent that incorporates (1) post-RL transformation of the design matrix, (2) addition of noise $\bm u_t$ at each gradient step, and (3) use of projection $\Pi_R(\cdot)$ on the response variable. 
The regular parameters, including $\eta$, $T$ and $C$ for the projected gradient method, are specified in Theorem~\ref{thm:ngd-bound} for the discussion of the accuracy of ${\bm{\hat\beta}^{\text{priv}}}$.
The injection of noise follows Lemma~\ref{lem:tighter-comp}. 
The scale of the Gaussian noise $\bm u_t$ at step $t$ depends on the privacy budget ($\epsilon, \delta$), and the noise scale factor $B$ associated with the sensitivity in the update function (\ref{eq:update}). 
The purpose of the projection on $\bm z$ 
is to bound the sensitivity of the gradient. With a proper choice of $R$ that scales up with $\sqrt{\ln n}$ (specified in Section \ref{sec:theoreical}), 
the projection does not affect the accuracy of the final estimator with high probability.

The major challenge lies in calculating the sensitivity. In the non-RL least square regression, two neighboring datasets $D = (X, \bm y)$ and $D = (X', \bm y')$ differ in a single row, making it straightforward to derive the sensitivity of the gradient of $\mathcal{L}_n(\bm \beta) =\frac{1}{2n}(\bm y - X\bm\beta)^\top (\bm y - X\bm\beta)$.
Here, in the context of post-RL analysis, we consider two neighboring datasets containing both linking variables and regression variables, denoted as $D = (X, \Phi_{X}, \bm y, \Phi_{\bm y})$ and $D' = (X', \Phi_{X'}, \bm y', \Phi_{\bm y'})$, which differ in the record of one individual. The change in one row of the quasi-identifiers $\Phi_{X}$ and $\Phi_{\bm  y}$ may affect more than one row of the matching probability matrix $Q$. As a result, the entries of the transformed design matrix $W = QX$ subject to change are not limited to one row as in the non-RL case. 
Consequently, determining the sensitivity of the gradient of $\mathcal{L}_n(\bm \beta) =\frac{1}{2n}(\bm z - W\bm\beta)^\top (\bm z - W\bm\beta)$ becomes non-trivial. This challenge distinguishes our work from \citet{Cai2019TheCO}. However, we will demonstrate in Section \ref{sec:theoreical} that, under a condition on the structure of $Q$, the sensitivity can be tracked.

\subsection{Post-RL Sufficient Statistics Perturbation}
\label{subsec:algo-ssp}
Noise can be injected into the process besides the gradient computation. Since the estimator interacts with the data through its (joint) sufficient statistics, an efficient way is to perturb the sufficient statistics to protect the data.
Such a technique, sufficient statistics perturbation (SSP), has been used in previous works such as \citet{Slavkovic2009, Foulds2016, wang2018}. For the 
{\color{black} 
non-private OLS estimator ${\hat {\bm \beta}^{\text{OLS}}}=(X^\top X)^{-1}X\bm y$}, to perturb the joint sufficient statistics $(X^\top X, X\bm y)$, it suffices to add noise to $A^\top A $ where $A = (X \mid \bm y)$ is the augmented matrix.
\citet{DworkTT014} offered an algorithm, ``Analyze Gauss'', to privately release $A^\top A $. It was later utilized by 
\citet{Sheffet17} for private linear regression, primarily perturbing the sufficient statistics. 

In our work, we adapt the ``Analyze Gauss'' algorithm to linear regression after record linkage, as shown in Algorithm \ref{alg:ssp}. 
The noise scale factor $B$ is the sensitivity of 
$A^\top A \stackrel{def}{=}\begin{pmatrix}
W^\top W & W^\top \bm z\\
\bm z ^\top W & \bm z ^\top \bm z
\end{pmatrix}$ which is specified in Section \ref{sec:theoreical}. 
The gram matrix $A^\top A$ exhibits properties that facilitate the computation of its sensitivity.
Algorithm \ref{alg:ssp} illustrates how incorporating the joint sufficient statistics in a comprehensive form facilitates the deployment of differential privacy.

\begin{algorithm}[ht]
\caption{Post-RL Sufficient Statistics Perturbation}
\begin{algorithmic}[1]
\Require  Linked dataset $(X,\bm z)$ and matching probability matrix $Q$, 
privacy budget ($\epsilon, \delta$), noise scale factor $B$, truncation level $R$.

\State Let $ W = QX$.
\State Generate a $d\times d$ symmetric Gaussian random matrix $U$ whose upper triangle entries (including the diagonal) are sampled i.i.d. from $\mathcal{N}(0, \omega^2)$ where $\displaystyle \omega = \frac{B \sqrt{2\ln(1.25/\delta)}}{\epsilon}$.
\State Generate a $d$-dimensional Gaussian random vector $\bm u$
whose entries are sampled i.i.d. from $\mathcal{N}(0, \omega^2)$.
{\color{black} 
\If{$W^\top W+U$ is computationally singular} 
    \State Repeat steps $2 \sim 3$.
\EndIf
}
\Ensure 
$\displaystyle {\bm{\hat\beta}^{\text{priv}}} = ( W^\top W + U )^{-1}(W^\top \bm z^* + \bm u)$ where $z^*= (\Pi_R(z_1), \dots, \Pi_R(z_n))^\top$.
\end{algorithmic}
\label{alg:ssp}
\end{algorithm}
\begin{remark}
\label{rmk:ssp}
{\color{black} 
    In step 4,  by post-processing,
    checking for singularity of $W^\top W+U$ consumes no extra privacy budget.
    In fact, the probability of $W^\top W+U$ being singular decreases exponentially as the sample size increases.  
}
\end{remark}

An alternative approach to implementing the SSP method is to add random noise separately to each sufficient statistic. In this approach, the total privacy budget should be divided between $X^\top X $ and $X\bm y$ for the estimation of linear regression, as proposed by \citet{wang2018}.
However, treating the joint statistics as a whole is more economical in terms of budgeting in general.
\citet{Lin2023} showed through comparison that splitting the total budget among the components results in introducing larger noise on average.
Although adding noise individually to the components of interest allows for the private release of each quantity, it is not part of the goal of the estimation.

\section{Theoretical Results}
\label{sec:theoreical}

In this section, we provide the theoretical results of the two algorithms introduced in Section \ref{sec:DP-algo}. The results are threefold: (1) differential privacy guarantees, (2) finite-sample error bounds, and (3) variances of the private estimators.
We present each of these along with a discussion of the corresponding conditions as they relate to the main variables in our record linkage model. All proofs for these results can be found in the supplementary materials.

\subsection{Privacy Guarantees}
The algorithms are designed to achieve certain privacy guarantees, given the corresponding sensitivity, for the post-RL case:
\begin{theorem}
[Privacy Guarantees]
Assume the following boundedness conditions hold:

(A1) There is a constant $c_{x} < \infty$ such that $\|\bm x\|_2 \leq c_{x}$.

(A2) Let $Q$ and $Q'$ be the matching probability matrices (MPMs) resulting from the neighboring datasets $D$ and $D'$ and let $Q\sim Q'$ denote such a relation. We assume that $\sup_{Q \sim Q'}\|Q - Q'\|_1 \leq M$ for some constant $M < \infty$, where $\|\cdot\|_1$ is the entry-wise 1-norm.  

{\color{black} Given the linked data $(X, \bm z)$ and the matching probability matrix $Q$ for the regression problem in (\ref{model:lm})}, under Assumptions (A1) and (A2), it follows that
\begin{enumerate}
    \item  Algorithm~\ref{alg:ngd} satisfies $(\epsilon, \delta)$-differential privacy with 
    \begin{equation}
        B = Rc_{x}  (M + 4) + 2Cc^2_{x}(M + 2),
    \end{equation} 
    
    \item Algorithm~\ref{alg:ssp} satisfies $(\epsilon, \delta)$-differential privacy with
    \begin{equation}
        B =  Rc_x(M+4) + \max\{2c_x^2(M+2), 2R^2\}.
    \end{equation}
\end{enumerate}
 
\label{thm:privacy}
\end{theorem} 
Essentially, we assume that the data domain is bounded,
which is critical for deriving a finite sensitivity of the target function on the data.
(A1) is a standard assumption for a \textit{bounded design} $X$. 
For the linking variables that are generally categorical, there are no analogous definitions of ``norm'' for numerical vectors. 
Instead, (A2) is imposed on 
the MPM since it summarizes all the information of the linking variables in the linkage model we consider.
Specifically, we assume that two MPMs produced by two neighboring datasets do not differ much in terms of the entry-wise 1 norm. This assumption characterizes a \textit{bounded linkage model}.

The rationale of (A2) is supported by typical schemes imposed on the structures of MPM in practice, as reviewed in Section \ref{subsubsec:rl-schemes}. 
For example, with the blocking scheme, the size of each block is manageably small (O(1)). When one record is altered, the fluctuation of the MPM is limited to at most two blocks.
Additionally, with the ELE model~(\ref{model:ele}),
as long as the changes to a single record only affect a finite number of records, the linkage accuracy $\gamma$ changes at most $O(1/n)$. Therefore, we have $\sup_{Q \sim Q'}\|Q - Q'\|_1 = O(1)$.
In general, a robust record linkage approach should not produce two considerably different MPMs from two neighboring datasets. 
Therefore, it is realistic to assume a bounded linkage model.

The proofs of Theorem~\ref{thm:privacy} revolve around calculating the sensitivity of the target function in each algorithm. 
Besides the upper bounds $c_{x}$ and $M$ discussed above, the sensitivity also depends on the truncation level $R$ on the response. Truncation is commonly used in DP algorithm designs when there are no priori bounds on the relevant quantities (e.g., \citet{Abadi2016}). In Section \ref{subsec:bounds}, we provide a specific choice of $R$ and present an accuracy statement with high probability.

\subsection{Finite-Sample Error Bounds}
\label{subsec:bounds}
We study the accuracy of the proposed estimators by deriving the finite-sample error bounds. In the following, we introduce two more assumptions in addition to (A1) and (A2):

(A3) The true parameter $\bm\beta$ satisfies $\|\bm\beta\|_2 \leq c_0$ for some constant $0 < c_0 < \infty$.

(A4) The minimum and maximum eigenvalues of $ W^\top  W/n$ satisfy 
\begin{equation}
0< \frac{1}{L} < d\lambda_{\min}\left(\frac{W^\top  W}{n} \right) \leq d\lambda_{\max}\left(\frac{W^\top  W}{n} \right)  < L
    \label{assump:eigen}
\end{equation}
 for some constant $1 < L < \infty$.

Assumption (A4) implies the smoothness and strong convexity of the loss function $\mathcal{L}_n(\bm \beta) = \frac{1}{2n}(\bm z - W\bm\beta)^\top (\bm z - W\bm\beta)$, which allows for a fast convergence rate for the gradient descent method in Algorithm \ref{alg:ngd}.
On the other hand, for Algorithm \ref{alg:ssp}, note that the term $(W^\top W)^{-1}$ is a component of sufficient statistics. Assumption (A4) offers a bound on the norm of $(W^\top W)^{-1}$, which helps derive the error bound of ${\bm{\hat\beta}^{\text{priv}}}$. 
{\color{black} 
Let Assumption (A4') be (A4) with $W$ replaced by $X$ and the constant $L$ replaced by $L'$. 
The larger of $L$ and $L'$ can be chosen as the constant to satisfy both (A4) and (A4'). 
Therefore, for convenience, we consider (A4) and (A4') to be the same assumption. 
}
\medskip
We first obtain the accuracy of the non-private estimators, for comparison purposes.

\begin{lemma}
\label{lem:olsest-error}
Let ${\hat {\bm \beta}^{\text{OLS}}} = \underset{\bm \beta}{\arg\min}(\bm y - X\bm\beta)^\top (\bm y - X\bm\beta)$ be the OLS estimator. Then,
    under (A4), it follows that $\displaystyle \mathbb{E}\|{{\hat {\bm \beta}^{\text{OLS}}} - \bm \beta}^{}\|^2  = \sigma^2 \operatorname{tr}(X^\top X)^{-1} = \Theta \left(\frac{\sigma^2d^2}{n}\right)$.
\end{lemma}

\begin{lemma} 
\label{lem:rlest-error}
Let ${\hat {\bm \beta}^{\text{RL}}} = \underset{\bm \beta}{\arg\min}(\bm z - W\bm\beta)^\top (\bm z - W\bm\beta)$ be the non-private record linkage estimator, and $\Sigma^{\text{RL}}$ be the covariance matrix of ${\hat {\bm \beta}^{\text{RL}}}$. Then, 
\begin{equation}
\label{eq:rlest-error}
    \mathbb{E}\|{{\hat {\bm \beta}^{\text{RL}}} - \bm \beta}^{}\|^2  =  \operatorname{tr}(\Sigma^{\text{RL}}),
\end{equation}
where $\Sigma^{\text{RL}} = (W^\top W)^{-1}W^\top \Sigma_{\bm z} W(W^\top W)^{-1}$.
\end{lemma}

As a special case, when the linkage is perfect (i.e., $Q$ is an identity matrix), the expected error of ${\hat {\bm \beta}^{\text{RL}}}$ in (\ref{eq:rlest-error}) takes the reduced form $\sigma^2 \operatorname{tr}(X^\top X)^{-1}$ which is exactly the lower bound obtained by ${\hat {\bm \beta}^{\text{OLS}}}$. Then, by Lemma \ref{lem:olsest-error},  we know that $\mathbb{E}\|{{\hat {\bm \beta}^{\text{RL}}} - \bm \beta}^{}\|^2$ is of order at least $\displaystyle\frac{\sigma^2d^2}{n}$ under (A4). From a secondary perspective regarding record linkage, it is beyond our scope to study how $\operatorname{tr}(\Sigma^{\text{RL}})$ behaves in general.

For the two proposed algorithms, we present upper bounds of the excess squared error of the private estimators, namely, $\|{\bm{\hat\beta}^{\text{priv}}} - {\hat {\bm \beta}^{\text{RL}}}\|^2$.
\begin{theorem}[Post-RL NGD]
\label{thm:ngd-bound}
{\color{black} Given the linked data $(X, \bm z)$ and the matching probability matrix $Q$ for the regression problem in (\ref{model:lm}),}
set the parameters of Algorithm \ref{alg:ngd}
as follows:
\begin{itemize}
    \item step size $\eta = d/L$, number of iterations $T = \lceil L^2 \ln(c_0^2n) \rceil$, feasibility $C = c_0$, initialization $\bm\beta^0=\bm 0$;
    \item truncation level $R = \sigma \sqrt{2\ln n}$;
    \item noise scale factor $B = Rc_{x}  (M+4) + 2c_0c^2_{x}(M+2)$;
\end{itemize}
   
Under Assumptions (A1)-(A4), given $\delta = o(1/n)$, with probability at least $1-c_1e^{-c_2\ln n} - e^{-c_3d}$ where $c_1, c_2, c_3$ are constants (see the proof), it follows that
\begin{equation}
    \|{\bm{\hat\beta}^{\text{priv}}} - {\hat {\bm \beta}^{\text{RL}}}\|^2 = \frac{1}{n} + O\left(\frac{\sigma^2 d^3 \ln^2n  \ln(1/\delta)}{n^2\epsilon^2} \right).
    \label{eq:privbound-ngd}
\end{equation}
\label{thm: priv}
\end{theorem}

\begin{theorem}[Post-RL SSP]
\label{thm:ssp-bound}
{\color{black} Given the linked data $(X, \bm z)$ and the matching probability matrix $Q$ for the regression problem in (\ref{model:lm})}, in Algorithm \ref{alg:ssp}, set
\begin{itemize}
    \item truncation level $R = \sigma \sqrt{2\ln n}$;
    \item noise scale factor $B =  Rc_x(M+4) + 2\max\{c_x^2(M+2), R^2\}$.
\end{itemize}
Under Assumptions (A1)-(A4), given $\delta = o(1/n)$, 
with probability at least $1-c_1e^{-c_2\ln n}-e^{-c_3d}$ where $c_1, c_2, c_3$ are constants (see the proof),
\begin{equation}
    \|{\bm{\hat\beta}^{\text{priv}}} - {\hat {\bm \beta}^{\text{RL}}}\|^2 = O\left(\frac{ \sigma^4 d^3 \ln^2n \ln(1/\delta)}{n^2\epsilon^2} \right).
    \label{eq:privbound-ssp}
\end{equation}
\label{thm: privanalyze}
\end{theorem}

In both algorithms, the response is projected with a level $R =  \sigma \sqrt{2 \ln n}$ where $\sigma^2$ is the homoskedastic variance of the random error in linear model (\ref{model:lm}). Let $\mathcal{E} = \{\Pi_R(z_i) = z_i, \forall i \in [n] \}$, then $\mathcal{E}$ is a high-probability event. The error bound is analyzed under $\mathcal{E}$, thus we obtain a statement with high probability.

In the NGD method, the bound consists of two parts on the RHS in (\ref{eq:privbound-ngd}). The first error term $1/n$ results from the convergence rate of gradient descent after $T$ iterations. The second error term is due to the addition of Gaussian noise for privacy and thus involves $\epsilon, \delta$. It is worth noting that the choice in theory $T = \lceil L^2\ln(c_0^2 n) \rceil$ is, to some extent, conservative to ensure the first error term is $O\left(1/n\right)$, which is the same order as $\mathbb{E}\|{{\hat {\bm \beta}^{\text{OLS}}} - \bm \beta}^{}\|^2$.
{\color{black} 
However, more iterations give rise to larger random noise being added to gradient updates due to a smaller privacy budget per iteration.
In practice, a smaller number of iterations may be favored for the tradeoff (see the experiment in Section \ref{subsec:app}), especially when $n$ is not sufficiently large.  
}

For the SSP algorithm, the convergence rate in (\ref{eq:privbound-ssp}) depends on similar variables as in the NGD algorithm. The major difference is that it is controlled by $\sigma^4$ instead of $\sigma^2$ due to the sensitivity of the gram matrix $A^\top A$ defined in Section \ref{subsec:algo-ssp}. However, the SSP method has a faster convergence rate when $n$ is sufficiently large. As a result, the SSP estimator is more susceptible to a large variance of the random error in the response variable whereas the NGD method is more robust. As we shall see in Section \ref{sec:numerical}, the performance of the two algorithms is different under various scenarios. 

Putting together Lemma \ref{lem:rlest-error} and Theorems \ref{thm:ngd-bound} and \ref{thm:ssp-bound}, we obtain a high probability error bound for each algorithm as follows.
\begin{corollary}
Under the regularity conditions (A1)-(A4),

(i) (Post-RL NGD)
\begin{equation}
    \mathbb{E}\|{\bm{\hat\beta}^{\text{priv}}} - {\bm \beta}\|^2 = O\left(\operatorname{tr}(\Sigma^{\text{RL}}) +   \frac{\sigma^2 d^3  \ln^2n \ln(1/\delta)}{n^2\epsilon^2}\right)
\end{equation}
with probability at least $1 - c_1e^{-c_2 \ln n} - e^{-c_3 d}$.

(ii) (Post-RL SSP)
\begin{equation}
     \mathbb{E}\|{\bm{\hat\beta}^{\text{priv}}} - {\bm \beta}\|^2 = O\left(\operatorname{tr}(\Sigma^{\text{RL}}) +   \frac{ \sigma^4 d^3  \ln^2n \ln(1/\delta)}{n^2\epsilon^2}\right)
\end{equation}
with probability at least $1 - c_1e^{-c_2 \ln n} - e^{-c_3 d}$.
\label{cor:final.bound}
\end{corollary}

\subsection{Variances}
\label{subsec:var}

{\color{black} 
As discussed in the Introduction, although a few works \citep{Alabi, Sheffet17} have addressed uncertainty of DP estimators through confidence bounds and intervals, the exact variance of DP estimators is rarely determined in most cases. 
}
Recent work, such as \citet{Lin2023}, has explored the variance of the private estimators for population proportions that have fairly simple structures.
The main barrier to the inspection of variance is that if the noise is injected into the intermediate steps  of the estimation process other than the output, then it is difficult to track the variability that noise introduces to the output estimator due to the  intricate nature of the algorithm.

The NGD and SSP algorithms are two examples where noise is added in the middle of the estimation process. The operations like function composition and taking the inverse complicate the inspection of the variance of the output estimator ${\bm{\hat\beta}^{\text{priv}}}$. 
To address this issue, we investigate the variance of ${\bm{\hat\beta}^{\text{priv}}}$ for the two algorithms by studying the variances of two proxy estimators. 
The theoretical variances of the proxy estimators can be used to approximate those of ${\bm{\hat\beta}^{\text{priv}}}$.

\begin{theorem}[Variance for Post-RL NGD]
\label{thm:var-ngd}
    In Algorithm \ref{alg:ngd} , if we consider the estimator without projections
\begin{equation}
\label{eq:est-noprojection}
    \bm{\beta}^{t+1} = \bm{\beta}^t - 
    \frac{\eta}{n}W^\top  (W \bm{\beta}^t - \bm z) + \bm u_t,
\end{equation}
then the variance of the $T$th iterate is given by
\begin{equation}
\label{var:ngd}
    \Sigma = \sum_{t=1}^T (I_d - A) ^{t-1} \cdot B^\top \Sigma_{\bm z} B  \cdot \sum_{t=1}^T (I_d - A) ^{t-1}+ \omega^2 \sum_{t=1}^T (I_d - A) ^{2t-2},
\end{equation}
where $I_{d}$ is the identity matrix of size $d$, $\displaystyle A \stackrel{def}{=} \frac{\eta}{n}W^\top  W$, $\displaystyle B \stackrel{def}{=} \frac{\eta}{n}W$, and $\omega^2$ is the variance of $\bm u_t$.

\end{theorem}

\begin{remark}
\label{rmk:var}
    In the non-private case where $\omega^2 = 0$, let $T \to \infty$, in which case
\begin{equation*}
    \Sigma \to A^{-1} B^\top \Sigma_{\bm z} BA^{-1} = (W^\top W)^{-1}W^\top \Sigma_{\bm z}W (W^\top W)^{-1} = \Sigma^{\text{RL}},
\end{equation*}
which is exactly the variance of ${\hat {\bm \beta}^{\text{RL}}}$ given in (\ref{eq:cov-rl}).
\end{remark}
The estimator in Algorithm \ref{alg:ngd} is a projected variant of (\ref{eq:est-noprojection}). 
The use of projection with level $C$ on $\bm\beta^{t+1}$ in (\ref{eq:update}) impedes the exact analysis of variance for ${\bm{\hat\beta}^{\text{priv}}}$. Instead, we provide the variance in (\ref{var:ngd}) for the non-projected estimator as a conservative variance for ${\bm{\hat\beta}^{\text{priv}}}$. The level of projection, the scale of noise, and the number of iterations together determine how conservative it is. From Remark \ref{rmk:var}, we know that as $T$ increases, the first term in the RHS of (\ref{var:ngd}) is getting close to $\Sigma^{\text{RL}}$. The second term, $\displaystyle\omega^2 \sum_{t=1}^T (I_d - A) ^{2t-2}$,  then summarizes the cumulative variability resulting from adding random noise at each iteration. Note that this term does not converge by simply increasing $T$, due to the fact that a smaller budget leads to larger noise at each iteration.

\begin{theorem}[Variance for Post-RL SSP]
\label{thm:var-ssp} 
For Algorithm~\ref{alg:ssp},
let $\hat{\bm\beta}' = {\hat {\bm \beta}^{\text{RL}}} + (W^\top W )^{-1} \bm u - ( W^\top W )^{-1} \cdot U(\bm {\hat {\bm \beta}^{\text{RL}}} + (W^\top W)^{-1} \bm u)$, 
{\color{black} then 
${\bm{\hat\beta}^{\text{priv}}} 
 - \hat{\bm\beta}' \stackrel{p}{\to} 0$ as $n \to \infty$}.
The variance of $\hat{\bm\beta}'$ is given by
\begin{equation}
    \Sigma = \Sigma^{\text{RL}} + \omega ^2 (W^\top W)^{-1} (I_d + \Sigma_0 + \Sigma_1 + \Sigma_2) (W^\top W)^{-1},
\end{equation}
where $\Sigma^{\text{RL}} = \operatorname{Cov}({\hat {\bm \beta}^{\text{RL}}})$ and the entries of $\Sigma_0$, $\Sigma_1$ and $\Sigma_2$ are given by
\begin{itemize}
    \item $(\Sigma_0)_{kk} = \sum_{i=1}^d\beta_i^2 $ for $k = 1, ..., d$;
    $(\Sigma_0)_{kl} = \beta_k\beta_l $ for $k \neq l$.
    \item $(\Sigma_1)_{kk} = \sum_{i=1}^d  \Sigma^{\text{RL}}_{ii}$ for $k = 1, ..., d$; 
    $(\Sigma_1)_{kl} = \Sigma^{\text{RL}}_{kl}$ for $k \neq l$.
    \item $ (\Sigma_2)_{kk} = \sum_{i=1}^d\Sigma'_{ii}$ for $k = 1, ..., d$;
    $(\Sigma_2)_{kl} = \Sigma'_{kl} $ for $k \neq l,$ where $\Sigma' \stackrel{def}{=} \omega^2(W^\top W)^{-2}$.
\end{itemize}
\begin{remark}
As shown in the proof of \ref{thm:var-ssp} (see the supplemental), the proxy estimator $\hat{\bm\beta}'$ is a first-order approximation for ${\bm{\hat\beta}^{\text{priv}}}$ using Taylor series for the term $(I + U(W^\top W)^{-1})^{-1}$ which appears in the decomposition of ${\bm{\hat\beta}^{\text{priv}}}$. 

\end{remark}

\end{theorem}
The variance of $\hat{\beta}'$ also consists of two parts: the variance of the non-private estimator ${\hat {\bm \beta}^{\text{RL}}}$ and the additional variation due to the noise injected for privacy purposes.
Given Assumption (A4), we have $\|(W^\top W)^{-1}\| = O(d/n)$ that appears in $\Sigma_1$ and $\Sigma_2$. As $n$ increases, the dominant component of the second term would be $\omega ^2 (W^\top W)^{-1} (I_d + \Sigma_0) (W^\top W)^{-1}$.

\section{Numerical Results}
\label{sec:numerical}
To evaluate the finite-sample performance of the proposed algorithms, we conduct a series of simulation studies and an application to a synthetic dataset that contains real data.

\subsection{Simulation Studies}
\label{subsec:simu}
In this section, we conduct simulation studies to assess the performance of the two proposed algorithms for simple linear regression with linked data. 
The non-private OLS estimator and RL estimator ${\hat {\bm \beta}^{\text{RL}}}$ by \citet{lahiri2005} are included as benchmarks. 
The private, non-RL counterpart methods are also performed in the absence of linkage errors for comparison.

For each simulation, a fixed design matrix $X$ and an matching probability matrix $Q$ are produced and
a total of 1000 repetitions are run over the randomness of both the intrinsic error $\bm e\sim \mathcal{N}(0, \sigma^2I_n)$ of the regression model and the noise injected for privacy. 
Figure~\ref{fig:sim} displays the $\ell_2$ relative error and both empirical and theoretical variances for the two settings. 

Two sets of simulations are conducted to explore the performance with varying sample size $n$ and $\sigma$, the homoskedastic variance of the random error in linear model (\ref{model:lm}). The parameters are set as follows:
\begin{itemize}
    \item ELE linkage model: blockwise linkage accuracy $\gamma_i$ characterizing $Q$,  block size $n_i = 25$.
    \begin{itemize}
        \item Settings 1 and 2: $\gamma_i \in \text{uniform}(0.6, 0.9)$, $M = 1$ in Assumption (A2).
        {\color{black} 
        \item Setting 3: the linkage accuracy $\gamma_i \equiv \gamma$ which varies from 0.6 to 1, while $M$ scales from 1 to 0. 
        }
    \end{itemize}
    \item regression model: $x_1,\dots,x_n \stackrel{i.i.d.}{\sim} \text{uniform}(-1,1)$, true regression coefficient $\beta = 1$. 
    \begin{itemize}
        \item Setting 1: $n$ varies from 3,000 to 10,000, $\sigma$ is fixed at 1.
        \item Setting 2: $n$ is fixed at $10,000$, $\sigma$ varies from 0.5 to 1.8.
        \item Setting 3: $n$ is fixed at $10,000$, $\sigma$ is fixed at 1.
    \end{itemize} 
    \item privacy budget: $\epsilon = 1$, $\delta = 1/n^{1.1}$.
\end{itemize}

\begin{figure}[H]
     \centering
     \begin{subfigure}[b]{0.495\textwidth}
         \centering
         \includegraphics[width=\textwidth]{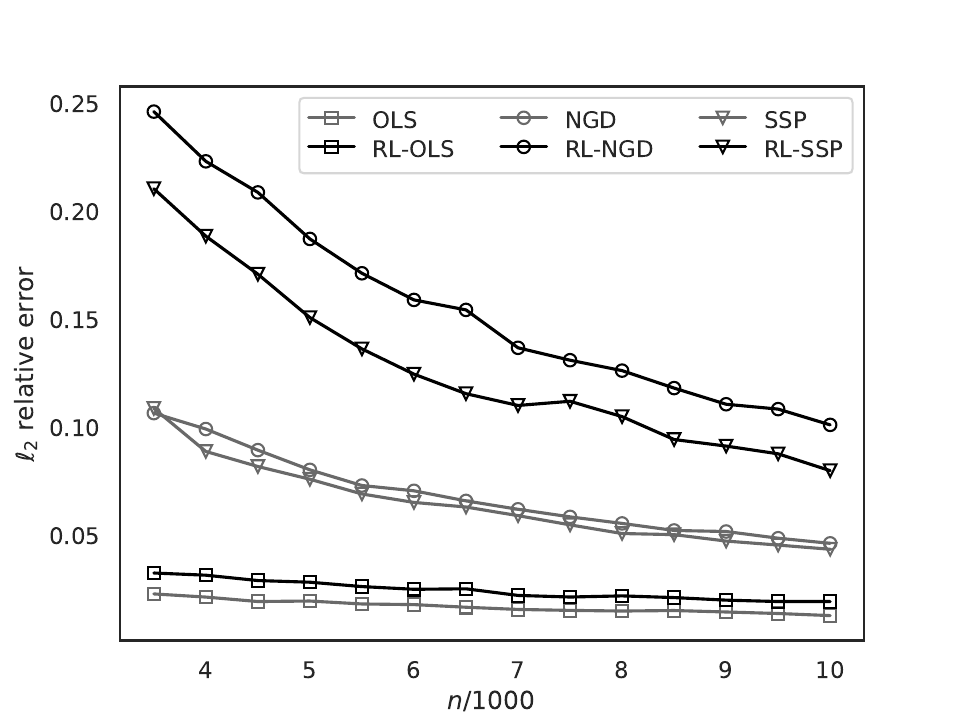}
         \caption{Setting 1: $\sigma = 1$}
         \label{fig:err_n}
     \end{subfigure}
     \hfill
     \begin{subfigure}[b]{0.495\textwidth}
         \centering
         \includegraphics[width=\textwidth]{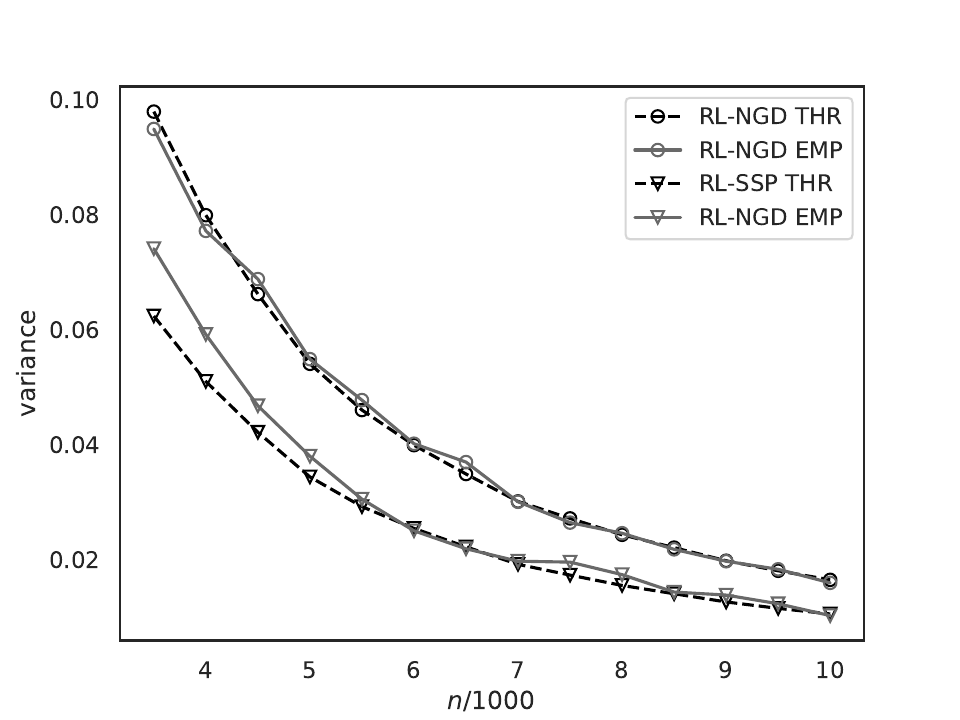}
         \caption{Setting 1: $\sigma = 1$}
         \label{fig:var_n}
     \end{subfigure}
     \hfill
\begin{subfigure}[b]{0.495\textwidth}
         \centering
         \includegraphics[width=\textwidth]{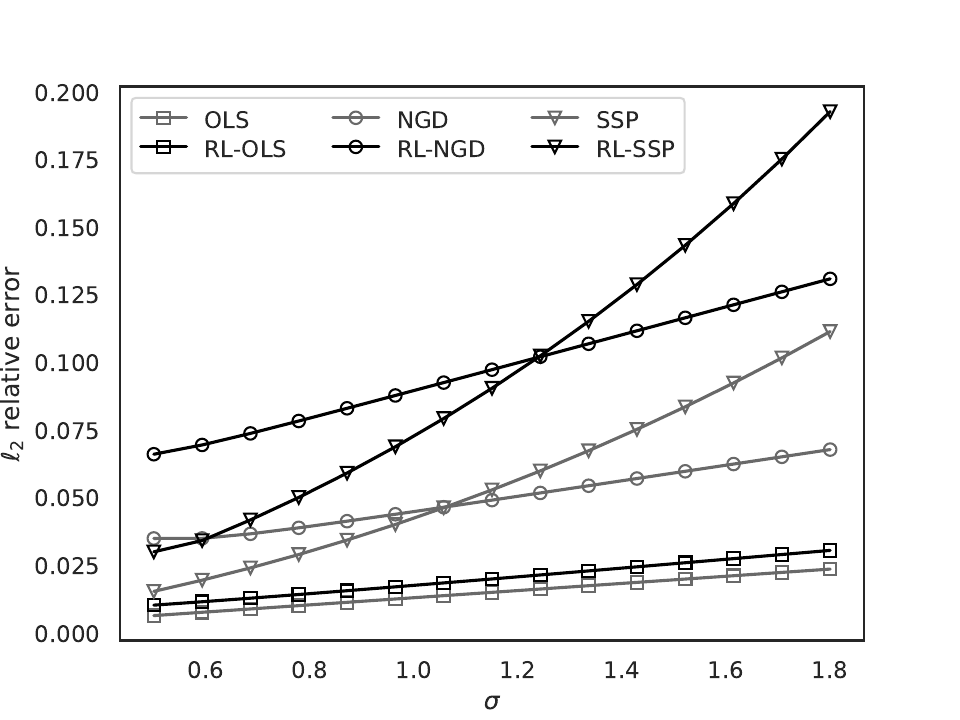}
         \caption{Setting 2: $n = 10,000$}
         \label{fig:err_sigma}
     \end{subfigure}
     \hfill
     \begin{subfigure}[b]{0.495\textwidth}
         \centering
         \includegraphics[width=\textwidth]{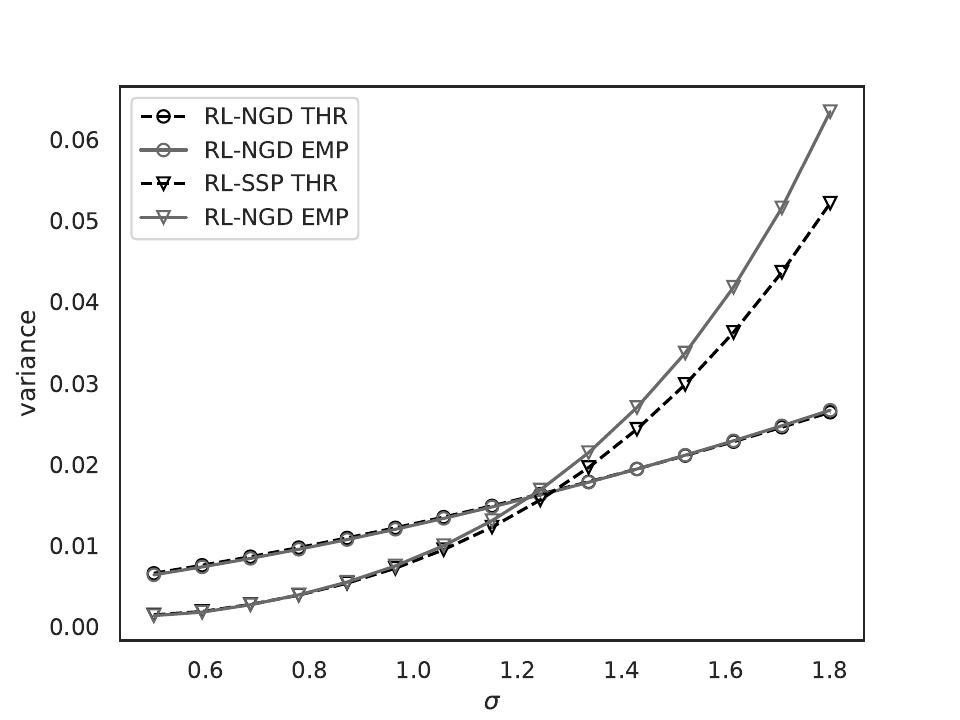}
         \caption{Setting 2: $n = 10,000$}
         \label{fig:var_sigma}
     \end{subfigure}

\hfill
\begin{subfigure}[b]{0.495\textwidth}
         \centering         \includegraphics[width=\textwidth]{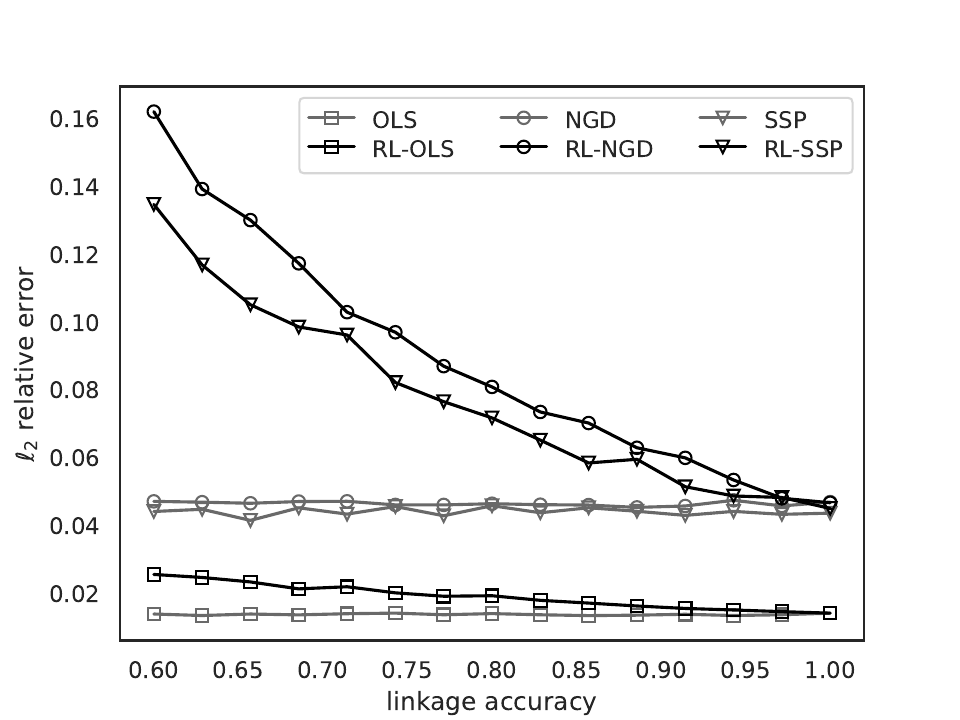}
         \caption{Setting 3: $\sigma = 1$, $n = 10,000$}
         \label{fig:err_lraccu}
     \end{subfigure}
     \hfill
     \begin{subfigure}[b]{0.495\textwidth}
         \centering
         \includegraphics[width=\textwidth]{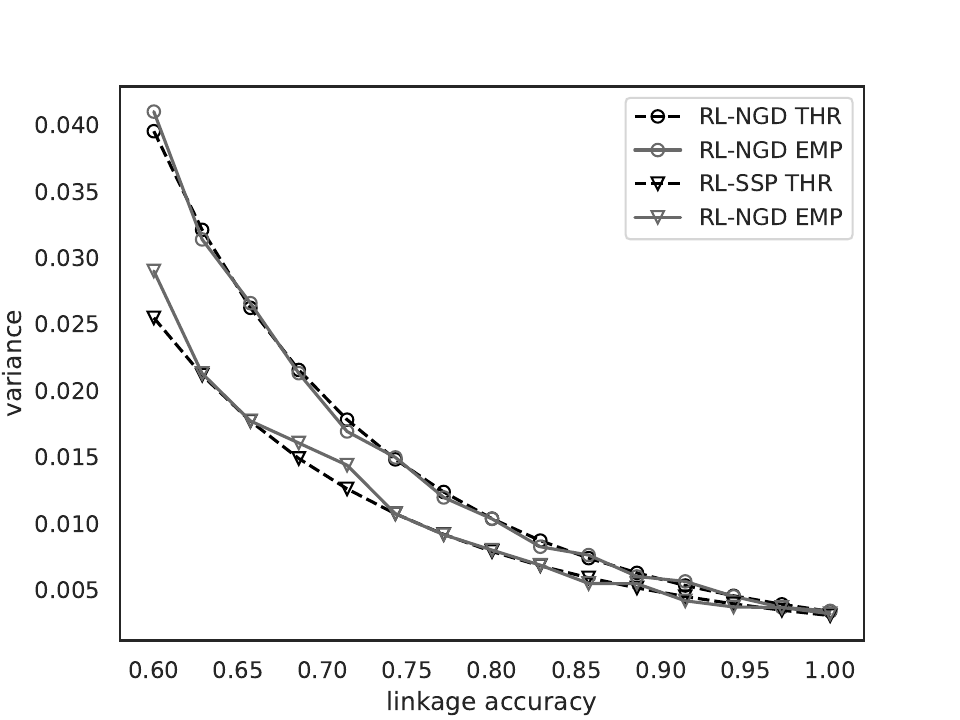}
         \caption{Setting 3: $\sigma = 1$, $n = 10,000$}
         \label{fig:var_lraccu}
     \end{subfigure}
     \caption{Average $\ell_2$-error and variance (theoretical versus empirical), with $(\epsilon, \delta) = (1, 8.5\times 10^{-5})$, against $n$ and $\sigma$, respectively.
     \\
     The ``RL-NGD'' and ``RL-SSP'' algorithms are our proposed post-RL approaches applied to the linked data, compared with the non-RL ``NGD'' and ``SSP'' methods applied to $(X, \bm y)$ (i.e., with no linkage errors).
     The non-private ``OLS'' and ``RL-OLS'' \citep{lahiri2005} results are also plotted for benchmarking.
     The number of iterations for ``RL-NGD'' results fall within the range of $(210, 260)$. }
     \label{fig:sim}
\end{figure}

In setting 1, where $\sigma$ is fixed at 1, Figure \ref{fig:err_n} shows the errors of all methods decrease 
with a growing sample size.
Due to the linkage errors,
the post-RL methods, including ${\hat {\bm \beta}^{\text{RL}}}$ and our two algorithms (denoted as ``RL-OLS'', ``RL-NGD'', and ``RL-SSP'' in the figures) run on the linked data $(X, \bm z)$, naturally always yield larger errors than their counterparts run on $(X, \bm y)$ when no linkage has to be done beforehand.
In this case, with $\sigma=1$, post-RL SSP outperforms post-RL NGD in terms of both accuracy and variance. 
However, as $\sigma$ increases, post-RL NGD algorithm starts perform better, as depicted in Figure \ref{fig:err_sigma} with varying $\sigma$. 
The error grows linearly for post-RL NGD and quadratically for post-RL SSP, which aligns with the theoretical results on the error bounds presented in Section \ref{subsec:bounds}.
Similar trends are observed for comparison of the non-RL NGD and SSP algorithms. 
{\color{black} 
In Figure \ref{fig:err_lraccu}, where linkage error tends to zero, the post-RL versions of the three estimators approach the corresponding non-RL versions. NGD and SSP methods have strictly larger error than OLS due to the cost of privacy.
}

Figures \ref{fig:var_n}, \ref{fig:var_sigma} and \ref{fig:var_lraccu} illustrate the empirical variances (EMP) against the theoretical variances (THR) of the proxy estimators given in Section \ref{subsec:var}. 
The theoretical variance of post-RL NGD closely aligns with the empirical variance at the chosen level of projection $C$. 
Recall that the theoretical variance would be exact when no projection is applied. 
Thus, with a lower level of projection on the gradient update, we anticipate it to be conservative.
On the other hand, the theoretical variance of post-RL SSP approximates well with moderately large $n$ and small $\sigma$. 
However, in scenarios with small $n$ and/or large $\sigma$, our theoretical variance may underestimate the reality due to the approximation's reliance on a first-order Taylor expansion. Therefore, one can include higher-order terms for better approximation.
{\color{black} 
In setting 3, where $n$ and $\sigma$ are fixed, as the linkage error vanishes, the variance reduces as a result of the smaller DP noise needed.
}

\subsection{Application to Synthetic Data}
\label{subsec:app}
Due to privacy concerns, pairs of datasets containing personal information, which serve as quasi-identifiers, are typically not made public.
We instead synthesize from a pair of generated quasi-identifiers datasets and real data for regression, as in \citet{Chambers2019SmallAE}.
For quasi-identifiers, we take advantage of the datasets generated by the Freely Extensible Biomedical Record Linkage (Febrl), which are available in the module \texttt{RecordLinkage} by \citet{rlpython} in Python. The pair of datasets for linkage we use correspond to 5000 individuals. The domain indicator (state) is used for blocking. The record linkage is performed based on the Jaro-Winkler score \citep{Jaro1989AdvancesIR} or exact comparison on 6 quasi-identifiers (names, date of birth, address, etc.). 
The maximum score is 6 for pairs that have exact alignment.
A threshold of 4 is chosen to link the records. For those left unlinked, we assign random links to ensure one-to-one linkage. 
A unique identifier is available in the dataset for verification purposes.
The resulting linkage accuracy for the 9 blocks are $\bm\gamma = (0.880,  0.903 , 0.918,  0.938, 0.905, 0.875, 0.898, 0.917, 0.898)$, making the overall accuracy 92.5\%. 
We adopt the ELE model for $Q$ and estimate it using $\bm\gamma$.

On the other hand, an anonymous dataset for regression comes from the Survey on Household Income and Wealth (SHIW) from the \citet{italydataset}. The net disposable income and consumption are the explanatory variable $X$ and the response $\bm y$, respectively. Since the SHIW dataset is larger, consisting of 8151 data points, we drop the outliers and randomly draw 5000 records and synthesize them with the Febrl dataset. 
{\color{black} Figure \ref{fig:syn_data} depicts the setup of the synthesization process. 
Using the unique identifier from the Ferlb dataset, the regression variables $(X, \bm y)$ are appended to the quasi-identifiers $(\Phi_{A}, \Phi_{B})$, resulting in two separate datasets: $(\Phi_{A}, X)$ and $(\Phi_{B}, \bm y)$.
Then, record linkage is performed by comparing $\Phi_{A}$ and $\Phi_{B}$ to output the linked data $(X, \bm z)$ and the matrix $Q$.  
}
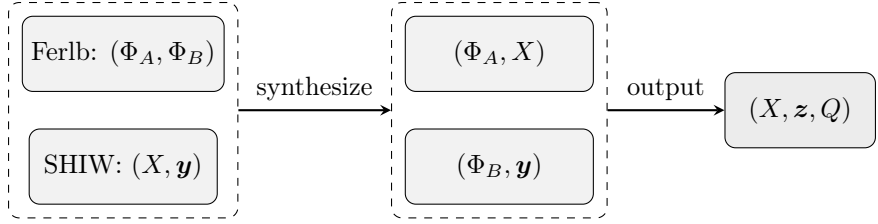
\begin{figure}[H]
\centering
\begin{tikzpicture}
    \node (ferlb) at (-5,1)  [rectangle, rounded corners, minimum width=2.5cm, minimum height=1cm, text centered, draw=black, fill=gray!10] {Ferlb: $(\Phi_{A}, \Phi_{B})$};
    \node (shiw) at (-5,-.5)  [rectangle, rounded corners, minimum width=2.5cm, minimum height=1cm, text centered, draw=black, fill=gray!10]  {SHIW: $(X, \bm y)$};
    
    \node (a) at (0,1)  [rectangle, rounded corners, minimum width=2.5cm, minimum height=1cm, text centered, draw=black, fill=gray!10] {$(\Phi_{A}, X)$};

    \node (b) at (0,-.5)  [rectangle, rounded corners, minimum width=2.5cm, minimum height=1cm, text centered, draw=black, fill=gray!10] {$ (\Phi_{B}, \bm y)$};
    
    \node (output) at (4,0.25)  [rectangle,  rounded corners, minimum width=2cm, minimum height=1cm, text centered, draw=black,  fill=gray!15] { $(X, \bm z, Q)$};

    \node [fit=(ferlb)(shiw), draw, inner sep=5pt, dashed, rounded corners,] (fs) {};

    \node [fit=(a)(b), draw, inner sep=5pt, dashed, rounded corners,] (ab) {};

     \draw [arrow] (fs) -- node [midway,above] {synthesize} (ab);
     
    \draw [arrow] (ab) -- node [midway,above] {output} (output);

\end{tikzpicture}
\caption{Synthesization.
The Ferlb dataset provides quasi-identifiers $(\Phi_{A}, \Phi_{B})$, and the SHIW dataset provides regression variables $ (X, \bm y)$.}
\label{fig:syn_data}
\end{figure}

To apply the proposed DP algorithms to the synthesized dataset, we set the (hyper)parameters as follows. The privacy budget is given by $(\epsilon, \delta) = (1, 8.5\times 10^{-5})$. 
The variance of the random error, $\sigma^2$, is estimated by the MSE.
The upper bounds in Assumptions (A1)-(A3) are set as: $M=1$, $c_0=1$, $c_x=\max(X) = 2.78$. 
In the NDG method, the projection level $C$ is set to 1.2.

To illustrate the importance of propagating linkage uncertainty when conducting downstream regression, we also apply the non-RL version of NGD and SSP algorithms. 
We obtain the non-RL regression results by running post-RL NGD and post-RL SSP methods with $M$ set to 0 and without converting $X$ into $W$.
This is equivalent to applying the non-RL methods discussed in \citet{Cai2019TheCO, Sheffet17} to the linked set $(X, \bm z)$ as if it were perfectly linked. 

\begin{figure}[ht]
    \centering
    \includegraphics[scale = 0.75
    ]{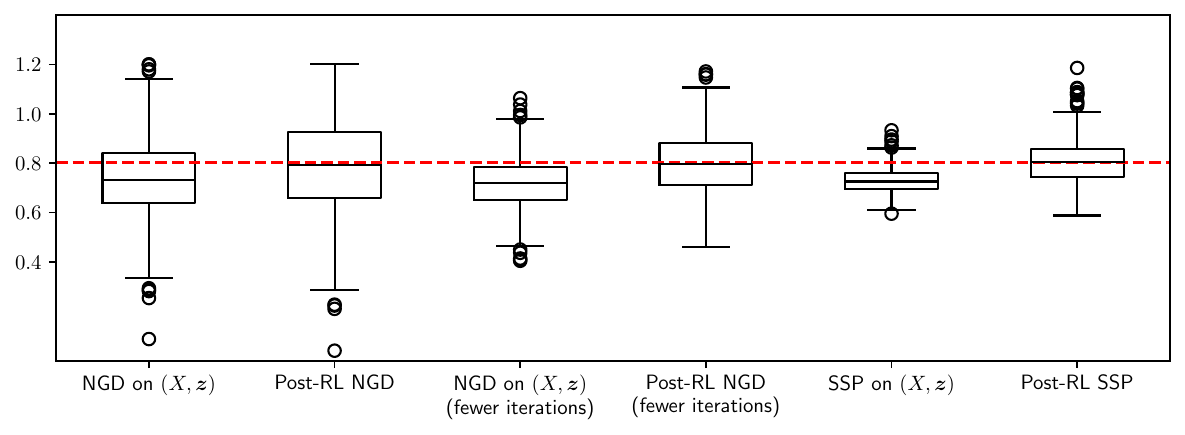} 
    \caption{Boxplots of DP estimates based on 1000 repetitions with $(\epsilon, \delta) = (1, 8.5\times 10^{-5})$. 
    \\
     {\color{black} The red dashed line indicates the OLS estimate. The proposed post-RL algorithms are compared with the non-RL ``NGD'' and ``SSP'' methods applied to $(X, \bm z)$ (i.e., without accounting present linkage errors). The third and fourth columns represent the two NGD methods running for $T = \lceil L^2\ln(c_0^2n)/3 \rceil$ iterations.}
    }
    \label{fig: boxplot}
\end{figure}

Figure \ref{fig: boxplot} 
displays the boxplots of the estimates of each algorithm. 
For each algorithm, a total of 1000 repetitions are done in order to reflect the randomness of the injected noise for privacy purposes.
The variables $X$ and $y$ are standardized before conducting simple linear regression. 
The OLS estimator on $(X, \bm y)$ (dashed line) is plotted for comparison. 
As can be seen, the DP estimators by running (non-RL) NGD and SSP on $(X, \bm z)$ directly are excessively biased as a consequence of ignoring linkage errors, even when the overall linkage accuracy is as high as 92.5\%.
Conversely, the results of post-RL NGD and post-RL SSP yields estimates centered around the OLS estimator but with higher variances, attributed to the cost of bias correction.
{\color{black} 
Post-RL NGD is more flexible due to hyperparameter tuning. Additionally, we run the NGD methods for fewer iterations with $T = \lceil L^2\ln(c_0^2n)/3 \rceil$, which is one-third of the value recommended by theory. We have found that this approach yields smaller variance while still producing accurate results in finite samples. Therefore, the theoretical number of iterations $T = \lceil L^2\ln(c_0^2n)\rceil$ may be conservative in some circumstances. Moderately reducing $T$ may lead to better results.
}

\section{Discussion}
\label{sec:discussion}

In this paper, we propose two differentially private algorithms for linear regression on a linked dataset that contains linkage errors, by leveraging the existing work on (1) linear regression after record linkage, and (2) differentially private linear regression. Figure \ref{Diagram} displays the connections among the related areas at a high level, including PPRL and SMPC mentioned in Sections \ref{sec:intro} and \ref{subsec:related-work}. Our work is the first one to simultaneously consider the linkage uncertainty propagation and the privatization of the output.
It also complements the area of PPRL where the main concern is the data leakage among different parties.
However, we do not discuss how to link the records in the first place and thus the security issues of the linkage process are beyond our scope. Instead, we treat record linkage from a secondary perspective: we begin with linked data prepared by an external entity and we have limited information about the linkage quality.

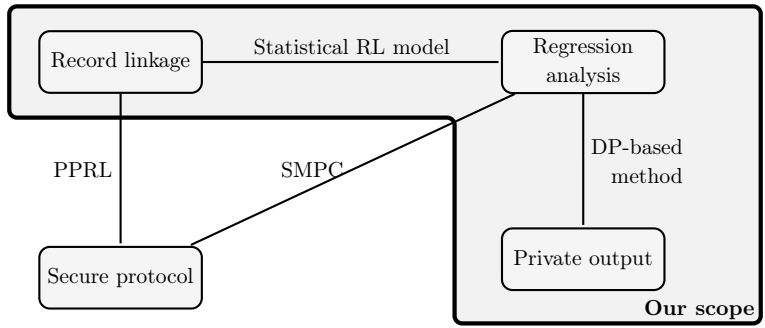
\begin{figure}[ht]
\centering
 \begin{tikzpicture}[%
  >=stealth,
  shorten >=1pt,
  node distance=2cm,
  on grid,
  auto,
  state/.append style={minimum size=2em},
  no arrow/.style={-,every loop/.append style={-}},
  thick,
  scale=0.74,
  every node/.style={scale=0.82}]
    \draw[rounded corners=1mm, black, fill=black!5, ultra thick] (-2,1) -- (11.5,1) -- (11.5,-4.7) -- (7,-4.7) -- (6, -4.7) -- (6, -1) -- (-2, -1) -- cycle;
    
    \node (RL)[box, label = below:] {Record linkage};
    \node (Model) [box, right of=RL, xshift=5.5cm, yshift=0cm] {Regression analysis};
    \node (security) [box, below of=RL, yshift=-1.5cm] {Secure protocol};
    \node (privacy) [box, below of=Model, yshift=-1.2cm] {Private output};
    \draw(10.4,-4.45) node {\textbf{Our scope}};
    \draw (RL) --node[above] {Statistical RL model} (Model) ;
    \draw (RL) --node[left] {PPRL} (security) ;
    \draw (Model) --node[right, align=right] {DP-based \\method} (privacy) ;
    \draw (Model) --node[left] {SMPC} (security) ;

\end{tikzpicture}
\caption{Diagram of related research areas. A secure protocol ensures no data is revealed to external parties during the linkage process.}
\label{Diagram}
\end{figure}
Specifically, we propose two post-RL algorithms
based on the noisy gradient descent and sufficient statistics perturbation methods from the DP literature. We provide privacy guarantees and finite-sample error bounds for these algorithms and discuss the variances of the private estimators. Our simulation studies and the application demonstrate the following: (1) the proposed estimators converge as the sample size increases; (2) post-RL linear regression incurs a higher cost than the non-RL counterpart in terms of the privacy-accuracy tradeoff; (3) The NGD method is flexible with hyperparameter tuning and can be applied to more general optimization problems; (4) SSP is specific to the least-squares problem, offering greater budget efficiency and more accurate results provided that the random error of the regression model is not too large.   

There are different directions to extend our work.
{\color{black} 
Note that there may be different scenarios of linking between the two datasets of the same set of entities. 
Assuming one-to-one linkage, as in our paper, is a canonical scenario. 
Although we do not explore it, we expect that our methods can be extended to other scenarios (e.g., one-to-many linkage) where $Q$ still makes sense. 
Extra assumptions may be required when determining the relevant sensitivities for privacy purposes.
}

One can also consider record linkage from a primary perspective. 
In addition to the traditional Fellegi–Sunter model, Bayesian approaches and machine learning-based methods have gained popularity. 
The record linkage may take forms other than the matching probability matrix adopted here.
Furthermore, when privacy concerns arise during the linkage process involving different parties, PPRL and SMPC protocols become essential. 
Tackling all the challenges depicted in Figure \ref{Diagram} simultaneously with a single efficient tool is of great practical use and significance. This interdisciplinary challenge requires expertise in both statistics and computer science.

Another important direction is exploring related statistical problems in the post-RL context, with or without privacy constraints. For example, confidence intervals and hypothesis testing are fundamental statistical inference tools.
Other potential problems that interest statisticians include high-dimensional linear regression and ridge regression. 





\appendix
\section{Proofs}\label{app}
\renewcommand{\thesection}{\Alph{section}}

\subsection{Lemmas}
\label{sec:lemmas}
The lemmas here support the proofs in Section \ref{sec:prfs}.

\begin{lemma}

If the minimum and maximum eigenvalues of $ W^\top  W/n$ satisfy $0< a < \lambda_{\min}( W^\top  W/n) \leq \lambda_{\max}( W^\top  W/n) < b$ for some constant $1 < L < \infty$, then the loss function
$\mathcal{L}_n(\bm \beta) = \frac{1}{2n}(\bm z - W\bm\beta)^\top (\bm z - W\bm\beta)$ is $b$-smooth and $a$-strongly convex.
\label{smoothconvex}
\end{lemma}

\begin{proof}
Since $\nabla \mathcal{L}_n(\bm \beta) = \frac{1}{n}(W^\top W\bm\beta - W^\top  \bm z)$, 
then for any $\bm \beta_1, \bm \beta_2$,
\begin{equation*}
    \begin{aligned}
        (\nabla \mathcal{L}_n(\bm \beta_1) - \nabla \mathcal{L}_n(\bm \beta_2))^\top (\bm \beta_1 - \bm \beta_2) 
        & = (\bm \beta_1 - \bm \beta_2)^\top \frac{W^\top W}{n} (\bm \beta_1 - \bm \beta_2)\\
        & \geq \lambda_{\min}\left(\frac{W^\top W}{n}\right)\|\bm \beta_1 - \bm \beta_2 \|^2 \\
        & \geq a\|\bm \beta_1 - \bm \beta_2  \|^2.
    \end{aligned}
\end{equation*}
By definition, $ \mathcal{L}_n(\bm \beta)$ is $a$-strongly convex.

For smoothness, we have 
\begin{equation*}
\begin{aligned}
     \|\nabla \mathcal{L}_n(\bm \beta_1) - \nabla \mathcal{L}_n(\bm \beta_2)\| & = \left\| \frac{W^\top W }{n} (\bm \beta_1 - \bm \beta_2)\right\| \\
     & \leq  \left\| \frac{W^\top W}{n} \right\| \|\bm \beta_1 - \bm \beta_2 \|  \\ 
     & = \lambda_{\max}\left(\frac{W^\top W}{n}\right)\|\bm \beta_1 - \bm \beta_2 \| \\
     & \leq b\|\bm \beta_1 - \bm \beta_2  \|.
\end{aligned}
\end{equation*}
The second equality is due to the fact that $\|A\| = |\lambda_{\max}(A)| $ for symmetric matrix $A$.
By definition, $ \mathcal{L}_n(\bm \beta)$ is $b$-smooth.

One can have a neater proof using 
alternative definitions (See Eq. (4) and (10) in \cite{lecture_smoothness2} for a twice differentiable function:

(1)$f$ is $\mu$–strongly convex if
$\lambda_{\min}(\nabla^2f) \geq \mu$;

(1)$f$ is $L$-smooth if
$\lambda_{\max}(\nabla^2f) \leq L$.

\end{proof}

\begin{lemma}[\cite{Bubeck2015}, Proof of Theorem 3.10.]
Let $f$ be $\alpha$-strongly convex and $\beta$-smooth on $\mathcal{X}$, and $x^*$ be the minimizer of $f$ on $\mathcal{X}$. Then
projected gradient descent with step size $\eta = \frac{1}{\beta}$ satisfies for $t \geq 0$,
\begin{equation*}
    \|x_{t+1} - x^*\|^2 \leq \left( 1- \frac{\alpha}{\beta}\right)\|x_{t} - x^*\|^2.
\end{equation*}

\label{conv}
\end{lemma}

\begin{lemma}
\label{lem:tri-ineq}
$\|\bm a+ \bm b\|^2 \leq (1+ c^2)\|\bm a \|^2+ (1+ 1/c^2)\|\bm b \|^2$ for any scalar $c \neq 0$.
\end{lemma}
\begin{proof}
Since
\begin{equation*}
    c^2\|\bm a\|^2 + \frac{1}{c^2}\|\bm b \|^2 - 2 \|\bm a \|\|\bm b\| = \left(\|c\bm a \| - \left\|\frac{1}{c}\bm b \right\|\right)^2 \geq 0,
\end{equation*}
it follows that
\begin{equation*}
    \|\bm a+ \bm b\|^2 \leq ( \|\bm a \| + \|\bm b\|) ^2 \leq \|\bm a \|^2 + \|\bm b \|^2 + 2 \|\bm a \|\|\bm b\| \leq  (1+ c^2)\|\bm a \|^2+ (1+ \frac{1}{c^2})\|\bm b \|^2.
\end{equation*}
\end{proof}

\begin{lemma} [\cite{Cai2019TheCO}, Lemma A.2]
For $X_1, ..., X_k \stackrel{i.i.d.}{\sim} \chi^2_d$, $\zeta>0$, $0<\rho<1$, 
$$
\mathbb{P}\left( \sum_{j=1}^k \zeta\rho^j X_j > \frac{\rho\zeta d}{1-\rho} + s\right) \leq \exp\left( -\min \left( \frac{(1 - \rho^2)s^2}{8\rho^2\zeta^2d}, \frac{s}{8\rho\zeta}\right)\right).
$$
\label{sum_of_chi}
\end{lemma}

\begin{lemma}[\cite{Sheffet17}, Proof of 
Proposition D.2]
For any invertible matrix $A$ and any matrix $B$ such that $(I + BA^{-1})$ is invertible, 
\begin{equation*}
    (A+B)^{-1} = A^{-1} - A^{-1}(I + BA^{-1})^{-1}BA^{-1}.
\end{equation*}
\label{lem: matrix}

\end{lemma}

\begin{lemma}
\label{lem:xy-var}
    Let X be a $d\times d$ symmetric random matrix with i.i.d upper triangle entries. Each entry has mean 0 and variance $\sigma^2$. Let $\bm y$ be a $d$-dimensional random vector, which has mean $\bm \mu$ and covariance matrix $\Sigma$. Let $\Sigma_{X\bm y}$ denote the covariance matrix of $X\bm y$. Then, the diagonal entries of $\Sigma_{X\bm y}$ are given by
    \begin{equation}
        (\Sigma_{X\bm y})_{kk} = \sigma^2 \sum_{i=1}^d(\mu_i^2 + \Sigma_{ii});
    \end{equation}
    the off-diagonal entries are
    \begin{equation}
        (\Sigma_{X\bm y})_{kl} = \sigma^2 (\mu_k\mu_l + \Sigma_{kl}) \text{ for } k \neq l.
    \end{equation}    
\end{lemma}

\begin{proof}
    Let $X = (\bm x_1, \bm x_2, ..., \bm x_d)$ where $\bm x_i = (x_{i1},x_{i2}, ..., x_{id})^T$ and $y = (y_{1},y_{2}, ..., y_{d})^T$.
    Then,
    \begin{equation}
        X\bm y = \sum_{i=1}^d \bm x_i y_i.
    \end{equation}
    Therefore,
    \begin{equation}
    \label{eq:var-xy}
        \operatorname{Var}(X\bm y) = \operatorname{Var}\left(\sum_{i=1}^d \bm x_i y_i\right) = \sum_{i=1}^d \operatorname{Var}( \bm x_i y_i) + \sum_{i\neq j}^d\operatorname{Cov}(\bm x_i y_i, \bm x_j y_j).
    \end{equation}
For the first term,
\begin{equation*}
\begin{aligned}
    \operatorname{Var}( \bm x_i y_i) & = \mathbb{E}[\operatorname{Var}( \bm x_i y_i \mid y_i)] + \operatorname{Var}[\mathbb{E}( \bm x_i y_i \mid y_i)] \\
    & = \mathbb{E}[y_i^2 \operatorname{Var}( \bm x_i  \mid y_i)] + \operatorname{Var}[y_i\mathbb{E}( \bm x_i \mid y_i)]\\
    & = \mathbb{E}(y_i^2) \operatorname{Var}(\bm x_i) + \operatorname{Var}(y_i) \mathbb{E}(\bm x_i) \\
    & = \sigma^2(\mu_i^2 + \Sigma_{ii})I_d.
\end{aligned}
\end{equation*}
For $i \neq j$,
\begin{equation*}
\begin{aligned}
        \operatorname{Cov}(\bm x_i y_i, \bm x_j y_j) & = \mathbb{E}[\bm x_i y_i(\bm x_j y_j)^T] - \mathbb{E}(\bm x_i y_i)\mathbb{E}(\bm x_j y_j) \\
        & = \mathbb{E}[(y_iy_j)\bm x_i \bm x_j ^T] - \mathbb{E}(\bm x_i)  \mathbb{E}(y_i)\mathbb{E}(\bm x_j) \mathbb{E}( y_j) \\
        & = \mathbb{E}(y_iy_j)\mathbb{E}(\bm x_i \bm x_j ^T),
\end{aligned}
\end{equation*}
where
\begin{equation*}
    \mathbb{E}(y_iy_j) = \mathbb{E}(y_i)\mathbb{E}(y_j) + \operatorname{Cov}(y_i, y_j) = \mu_i\mu_j + \Sigma_{ij},
\end{equation*}
and $\mathbb{E}(\bm x_i \bm x_j ^T)$ is a $d \times d$ matrix with the $(j,i)$ entry being $\sigma^2$ and 0 other wise. 
Putting $\operatorname{Var}( \bm x_i y_i)$ and $\operatorname{Cov}(\bm x_i y_i, \bm x_j y_j)$ back in (\ref{eq:var-xy}),
we know that 
$\operatorname{Var}( X \bm y)$ has the diagonal entries $$\sigma^2 \sum_{i=1}^d(\mu_i^2 + \Sigma_{ii})$$ and the $(k,l)$ off-diagonal entry $$\sigma^2 (\mu_k\mu_l + \Sigma_{kl}).$$  
\end{proof}

\begin{remark}
    In Lemma~\ref{lem:xy-var},
    $\operatorname{Var}( X \bm y)$ is given by $\sigma^2(\bm \mu \bm \mu ^T + \Sigma )$ with the diagonal entries replaced by its trace.
\end{remark}

\begin{lemma}
\label{lem:cov1}
    Let $X$ be a $d\times d$ random matrix with $ \mathbb{E} X = 0_{d\times d}$. Let $\bm y$ be a $d$-dimensional random vector that is independent of $X$.
    Then, $
        \operatorname{Cov}(\bm y, X\bm y) = 0_{d\times d}
   $.
\end{lemma}
\begin{proof}
\begin{equation*}
    \begin{aligned}
    \operatorname{Cov}(\bm y, X\bm y)
    & = \mathbb{E}[(\bm y - \mathbb{E} \bm y)(X\bm y - \mathbb{E} X \bm y)^\top]\\
    & = \mathbb{E}[(\bm y - \mathbb{E} \bm y)(X\bm y - \mathbb{E} X \mathbb{E} \bm y)^\top]\\
    & = \mathbb{E}[(\bm y - \mathbb{E} \bm y)(X\bm y)^\top]\\
    & = \mathbb{E}\bm y \bm y^\top \mathbb{E} X^\top - \mathbb{E} \bm y \mathbb{E} \bm y^\top \mathbb{E} X^\top  = 0_{d\times d}.
\end{aligned}
\end{equation*}

\end{proof}

\begin{lemma}
\label{lem:cov2}
    Let $X$ be a $d\times d$ random matrix with $ \mathbb{E} X = 0_{d\times d}$. Let $\bm y$ be a $d$-dimensional random vector. Let $\bm z$ be another $d$-dimensional random vector that is independent of both $X$ and $\bm y$ and $ \mathbb{E} \bm z = 0_{d}$. 
    Then, $
        \operatorname{Cov}(X\bm y, X\bm z) = 0_{d\times d}
    $.
\end{lemma}
\begin{proof}
\begin{equation*}
    \begin{aligned}
    \operatorname{Cov}(X\bm y, X\bm z)
    & = \mathbb{E}[(X\bm y - \mathbb{E}X \bm y)(X\bm z - \mathbb{E} X \bm z)^\top]\\
    & = \mathbb{E}[(X\bm y - \mathbb{E}X \mathbb{E}\bm y)(X\bm z - \mathbb{E} X \mathbb{E}\bm z)^\top]\\
    & = \mathbb{E}[X \bm y \bm z^\top X^\top].
\end{aligned}
\end{equation*}
Since $\bm z$ is independent of both $X$ and $\bm y$ and the entries of $\bm z$ appear linearly in every entry of $X \bm y \bm z^\top X^\top$. By the zero-expectation of $\bm z$,
$ \operatorname{Cov}(X\bm y, X\bm z) = 0_{d\times d}$.

\end{proof}

\subsection{Proofs}
This section provides the proofs for theorems and lemmas presented in the paper.
\label{sec:prfs}

To derive this tighter composition for $(\epsilon,\delta)$-DP, we utilize the notion of zero-concentrated differential privacy (zCDP, \cite{Bun2016ConcentratedDP}), defined as follows.

\begin{definition}[$\rho$-zCDP] 
A randomized mechanism $M: \mathcal{X}^n \rightarrow \mathcal{ Y}$ is $\rho$-zero-concentrated-differentially private ($\rho$-zCDP) if, for
all $\bm x \sim \bm x' \in \mathcal{ X}^n$ differing on a single entry and all $\alpha \in (1, \infty)$,
 \begin{equation*}
    \operatorname{D}_\alpha (M(\bm x) \|M(\bm x')) \leq \rho\alpha,
    \label{defzcdp}
 \end{equation*}
where $ \operatorname{D}_\alpha (M(\bm x) \|M(\bm x'))$ is the $\alpha$-Rényi divergence \cite{van2014renyi}
between the distribution of $M(\bm x)$ and
the distribution of $M(\bm x')$. 
\label{def:rhozcdp}
\end{definition}

\begin{proof} [Proof of Lemma~\ref{lem:tighter-comp}]   
    \cite{Bun2016ConcentratedDP} have shown that, like the classic $(\epsilon, \delta)$-DP notion, 
    $\rho$-zCDP enjoys properties including basic composition and post-processing. 
    The corresponding Gaussian mechanism (Proposition 1.6, \cite{Bun2016ConcentratedDP}) states that an algorithm $f$ is $\rho$-zCDP after adding Gaussian noise $\mathcal{N}\left(0, \frac{\Delta_f^2}{2\rho}\right)$ to it.
    In addition, they have shown that $\rho$-zCDP implies $(\rho+2\sqrt{\rho\ln(1/\delta}), \delta)$-DP (Proposition 1.3, \cite{Bun2016ConcentratedDP}).

    Therefore, to achieve $(\epsilon, \delta)$-DP, it suffices for the algorithm to be $\rho$-zCDP with $\rho:= \epsilon+2\ln(1/\delta)-2\sqrt{(\epsilon+\ln(1/\delta))\ln(1/\delta)}$.
    Using the Gaussian mechanism and basic composition rule of $\rho$-zCDP, it suffices to add noise 
    \begin{equation*}
        u_t \sim \mathcal{N}\left(0, \frac{T\Delta_t^2}{2\rho}\right)
    \end{equation*}
    to $f_t$ for all $t$.
    Note that if $\epsilon \leq \frac{8\ln(1/\delta)}{2+\sqrt{2}}$, then $\rho \geq \frac{\epsilon^2}{8\ln(1/\delta)}$.
    Therefore, it suffices to add noise 
    \begin{equation*}
        u_t \sim \mathcal{N}\left(0, \frac{4T\Delta_t^2\ln(1/\delta)}{\epsilon^2}\right)
    \end{equation*}
    to $f_t$ for all $t$.

\end{proof}

\begin{proof}[Proof of Theorem~\ref{thm:privacy}(i)]
\label{prf:privay-ngd}
By the composition proposition of differential privacy, to establish that Algorithm \ref{alg:ngd} (Post-RL Noisy Gradient Descent) satisfies $(\epsilon,\delta)$-differential privacy it suffices to show that the computation of  $\beta^{t+1}$ is $(\epsilon/T, \delta/T)$-differentially private. According to the Gaussian mechanism 
Theorem 2.1, showing the latter boils down to proving that the sensitivity is controlled at each gradient step. Let $g^{t+1}(\bm \beta^t ; X, \bm y, Q) = \sum_{i=1}^n(\bm w_i^\top \bm\beta^t  - \Pi_R(z_i))\bm w_i$. We will show that the $\ell_2$-sensitivity of $g^{t+1}$, denoted by $\Delta_g$, is bounded by $B$.

Without loss of generality, we assume the neighboring data sets $(X, \Phi_{X}, \bm y, \Phi_{\bm y}) $ and $(X',\Phi_{X'}, \bm y', \Phi_{\bm y'})$ differ in the $j$-th record. Recall $\bm z$ is a permutation of $\bm y$ satisfying $q_{ij} = P(z_i = y_i)$. Let $\bm z'$ be a copy of $\bm z$ but with the entry $y_j$ changed to $y'_j$. 
Recall in the record linkage linear model elaborated in \cite{lahiri2005}, 
$\bm w_i = \sum_{j=1}^nq_{ij}\bm x_j$, which are convex combinations of rows of $X$.
We can write
$\bm w_i = \sum_{k \neq j} q_{ik} \bm x_k + q_{ij} \bm x_j$ and $\bm w'_i = \sum_{k \neq j} q'_{ik} \bm x_k + q'_{ij} \bm x'_j$. From assumption (A1) that $\|\bm x\|<c_{\bm x}$ with probability 1, we know $\|\bm w_i\|<c_{\bm x}$ almost surely.
Then, we have
\begin{equation}
\begin{aligned}
    \|\bm w_i - \bm w'_i\| & = \left\|\left(\sum_{k \neq j} q_{ik} \bm x_k + q_{ij} \bm x_j \right) - \left(\sum_{k \neq j} q'_{ik} \bm x_k + q'_{ij} \bm x'_j \right)\right\| \\
     & = \left\| \sum_{k \neq j} (q_{ik} - q'_{ik}) \bm x_k + (q_{ij}  \bm x_j - q'_{ij} \bm x_j') \right \| \\
     & = \left\| \sum_{k \neq j} (q_{ik} - q'_{ik}) \bm x_k  + (q_{ij} - q'_{ij}) \bm x_j +  q'_{ij} (\bm x_j - \bm x_j') \right \| \\
     & \leq \left\| \sum_{k = 1}^n (q_{ik} - q'_{ik}) \bm x_k \right\| + \left\| q'_{ij} (\bm x_j - \bm x_j') \right \|\\
     & \leq c_{\bm x}\left(  \sum_{k = 1}^n |q_{ik} - q'_{ik}| +  2q'_{ij} \right)  . 
\end{aligned}
\end{equation}

Since $Q'$ is a doubly stochastic matrix,  $\sum_{i=1}^n q'_{ij} = 1$ for any $j$. By the arbitrariness of index $j$, it follows that 
\begin{equation}
    \begin{aligned}
        \max_{(X, \bm y, Q)\sim(X', \bm y', Q')} \sum_{i=1}^n \|\bm w_i - \bm w'_i\|  \leq \sum_{i=1}^n c_{\bm x} \left( \sum_{k = 1}^n |q_{ik} - q'_{ik}| +  2q'_{ij} \right)  
        = c_{\bm x}(\|Q - Q'\|_1 + 2).
    \end{aligned}
\label{eq:wiwiprime}
\end{equation}

The sensitivity of $g^{t+1}$ is
\begin{equation}
\begin{aligned}
\Delta_g & = \max_{(X, \bm y, Q)\sim(X', \bm y', Q')}   \left\|g^{t+1}(\bm\beta^{t};X, \bm y, Q) - g^{t+1}(\bm\beta^{t};X',  \bm y', Q')\right\| \\
& = \max_{(X, \bm y, Q)\sim(X', \bm y', Q')}  \left\| \sum_{i=1}^n(\bm w_i^\top \bm\beta^t  - \Pi_R(z_i))\bm w_i - \sum_{i=1}^n(\bm w_i^{\prime T}\bm\beta^t  - \Pi_R(z'_i))\bm w'_i \right\|\\
& \leq \max_{(X, \bm y, Q)\sim(X', \bm y', Q')}   \sum_{i=1}^n \left\|\bm w_i^\top \bm\beta^t \bm w_i - \bm w_i^{\prime T}\bm\beta^t  \bm w'_i \right\| + \max_{(X, \bm y, Q)\sim(X', \bm y', Q')}    \sum_{i=1}^n \left\| \Pi_R(z_i)\bm w_i -  \Pi_R(z'_i)\bm w'_i \right\|\\
& =: \Delta_1 + \Delta_2.
\end{aligned}
\label{sens: g}
\end{equation}

We use $\Delta_1$ and $\Delta_2$ to denote the two terms on the right of (\ref{sens: g}), respectively.
To bound the first term $\Delta_1$, since
\begin{equation*}
\begin{aligned}
    \|\bm w_i^\top  \bm \beta^t  \bm w_i - \bm w_i^{\prime T}\bm \beta^t  \bm w_i^{\prime} \|
    & = \|\bm w_i^\top  \bm \beta^t  (\bm w_i - \bm w_i^{\prime} + \bm w_i^{\prime})- \bm w_i^{\prime T}\bm \beta^t  \bm w_i^{\prime} \|\\
    & =  \|\bm w_i^\top  \bm \beta^t  (\bm w_i - \bm w_i^{\prime}) + \bm w_i^\top  \bm \beta^t \bm w_i^{\prime}- \bm w_i^{\prime T}\bm \beta^t  \bm w_i^{\prime} \|\\
    & = \|\bm w_i^\top  \bm \beta^t  (\bm w_i - \bm w_i^{\prime}) + (\bm w_i - \bm w_i^{\prime })^T\bm \beta^t  \bm w_i^{\prime} \|\\
    & \leq \|\bm w_i^\top  \bm \beta^t  (\bm w_i - \bm w_i^{\prime})\| + \|(\bm w_i - \bm w_i^{\prime })^T\bm \beta^t  \bm w_i^{\prime} \|\\
    & \leq \|\bm w_i\| \|\bm \beta^t \| \|\bm w_i - \bm w_i^{\prime }\| + \|\bm w_i - \bm w_i^{\prime }\| \|\bm \beta^t \| \|\bm w_i^{\prime}\|,
\end{aligned}
\end{equation*}
then by (\ref{eq:wiwiprime}), $\Delta_1$ can be controlled:
\begin{equation}
\begin{aligned}
    \Delta_1 & \leq  \max_{(X, \bm y, Q)\sim(X', \bm y', Q')} \sum_{i=1}^n 2Cc_{\bm x} \|\bm w_i - \bm w_i^{\prime }\|
     \leq 2Cc^2_{\bm x}(\|Q - Q'\|_1 + 2) \leq 2Cc^2_{\bm x}(M + 2).
\end{aligned}
\end{equation}

For the second term $\Delta_2 =  \max_{(X, \bm y, Q)\sim(X', \bm y', Q')}  \sum_{i=1}^n\|\Pi_R(z_i)\bm w_i - \Pi_R(z'_i)\bm w'_i\|$,
since
\begin{equation}
\begin{aligned}
    \|\Pi_R(z_i)\bm w_i - \Pi_R(z'_i)\bm w'_i\|
    & = \|\Pi_R(z_i)(\bm w_i - \bm w'_i) + (\Pi_R(z_i) - \Pi_R(z'_i))\bm w'_i\|\\
    & \leq \|\Pi_R(z_i)(\bm w_i - \bm w'_i)\| + \|(\Pi_R(z_i) - \Pi_R(z'_i))\bm w'_i\|\\
    & \leq R\|(\bm w_i - \bm w'_i)\| + c_{\bm x}\|\Pi_R(z_i) - \Pi_R(z'_i)\|.
\end{aligned}
\end{equation}
Then, 
\begin{equation}
\begin{aligned}
    \Delta_2 & \leq  \max_{(X, \bm y, Q)\sim(X', \bm y', Q')}\sum_{i=1}^n \left(R\|(\bm w_i - \bm w'_i)\| + c_{\bm x}\|\Pi_R(z_i) - \Pi_R(z'_i)\| \right) \\
    & = (Rc_{\bm x}(\|Q - Q'\|_1 + 2) + 2R c_{\bm x}).  \\ 
    & \leq Rc_{\bm x} (M + 4).  
\end{aligned}
\end{equation}

It follows that
\begin{equation*}
    \Delta_g \leq \Delta_1 + \Delta_2 \leq 2Cc^2_{\bm x}(M + 2) + Rc_{\bm x} (M + 4).
\end{equation*}
\end{proof}

\begin{proof}[Proof of Theorem~\ref{thm:privacy}(ii)]
We now show that Algorithm~\ref{alg:ssp} (Post-RL sufficient statistics perturbation) is $(\epsilon,\delta)$-differentially private.  Let $A = (W \mid \bm z)$ be the augmented matrix considering linkage errors. Then $A^\top A$ contains sufficient statistics for the $\bm \beta$. Thus, it suffices to show that the sensitivity of  $A^\top A$ is controlled by $B = Rc_x(M+4) + \max\{2c_x^2(M+2), 2R^2\}$.

Let $A' = (W' \mid \bm z')$, where $W'$ and $\bm z'$ come from any neighboring data set, as in the proof of Theorem~\ref{thm:privacy} (i). 
We have
\begin{equation*}
\begin{aligned}
    A^\top A - A'^\top A' & = 
    \begin{pmatrix}
    W^\top W - W'^\top W' & 0\\
    0 & \bm z^\top\bm z - \bm z'^\top \bm z'
    \end{pmatrix}
    +
    \begin{pmatrix}
    0 & W^\top \bm z - W'^\top \bm z'\\
    \bm z^\top W - \bm z'^\top W' & 0
    \end{pmatrix}\\
    & =: A_1 + A_2. 
\end{aligned}
\end{equation*}
By the properties of the norm of block matrices, 
\begin{equation*}
    \|A_1\| \leq \max\{\|W^\top W - W'^\top W' \|, \| \bm z^\top\bm z - \bm z'^\top \bm z'\|\}
\end{equation*}
and
\begin{equation*}
\begin{aligned}
    \|A_2\| & = \|W^\top \bm z - W'^\top \bm z'\| \\
    & = \|W^\top \bm z - W'^\top \bm z + W'^\top \bm z - W'^\top \bm z' \|\\
    & \leq \|(W-W')^\top\bm z\| + \|W'^\top(\bm z - \bm z')\|.
\end{aligned}
\end{equation*}
We have
\begin{equation*}
\begin{aligned}
    \|W^\top W - W'^\top W' \|
    & = \|\sum_{i}^n\bm w_i \bm w_i^\top - \sum_{i}^n \bm w'_i \bm w'^\top_i\| \\
    & \leq \sum_{i}^n \| \bm w_i(\bm w_i - \bm w_i')^\top + (\bm w_i-\bm w_i')\bm w_i'^\top\| \\
    & \leq 2 c_x \sum_{i}^n \|\bm w_i - \bm w_i'\| \\
    & \leq 2 c_x^2(M+2),
\end{aligned}
\end{equation*}
and
\begin{equation*}
    \|\bm z^\top\bm z - \bm z'^\top \bm z'\| = |y^2_j - y'^2_j| \leq 2R^2.
\end{equation*}
Note that we can swap the rows of $\bm z'$ and $Q'$, such that $\bm z$ and $\bm z'$ only differ in one record and it does not change the estimation using $Q'$ and $\bm z'$ after swapping. 
Then,
\begin{equation*}
    \|A_1\| \leq \max\{2c_x^2(M+2), 2R^2\}.
\end{equation*}
Since
\begin{equation*}
    \|(W-W')^\top\bm z\| = \|\sum_{i}^n (\bm w_i - \bm w'_i) z_i\| \leq R \sum_{i}^n \|\bm w_i - \bm w_i'\| \leq Rc_x(M+2)
\end{equation*}
and 
\begin{equation*}
    \begin{aligned}
        \|W'^\top(\bm z - \bm z')\| 
        & = \|\sum_{i}^n \bm w'_i(z_i - z_i') \| 
        & = \sum_{i}^n \|\bm w_i\| |z_i - z_i'| 
        & \leq c_x \sum_{i}^n |z_i - z_i'|
        & \leq 2Rc_x,
    \end{aligned}
\end{equation*}
we have
\begin{equation*}
    \|A_2\| \leq Rc_x(M+4).
\end{equation*}
Putting the upper bounds together,
we derive
\begin{equation*}
    \max_{A,A'}\|A^\top A - A'^\top A'\| \leq Rc_x(M+4) + \max\{2c_x^2(M+2), 2R^2\}.
\end{equation*}

\end{proof}

\begin{proof}[Proof of Lemma~\ref{lem:olsest-error}]
\label{Prf:olsest}
To establish the behavior of the expected squared-error loss of ${\hat {\bm \beta}^{\text{OLS}}}$, first note that since 
\begin{equation*}
0< \frac{1}{L} < d\lambda_{\min}\left(\frac{X^\top  X}{n} \right) \leq d\lambda_{\max}\left(\frac{X^\top  X}{n} \right)  < L,
\end{equation*}
we have
\begin{equation*}
    \frac{d}{Ln} < \lambda_{\max}( X^\top  X)^{-1} \leq \lambda_{\max}( X^\top  X)^{-1} < \frac{dL}{n}.
\end{equation*}
Therefore, $\operatorname{tr}(X^\top X)^{-1} = \sum_{i=1}^d \lambda_i ((X^\top X)^{-1}) =  \Theta(d^2/n)$ where $\lambda_i, \quad i = 1, \dots, d$ are the eigenvalues of $(X^\top X)^{-1}$.
\end{proof}

\begin{proof}[Proof of Lemma~\ref{lem:rlest-error}]
\label{Prf:rlsigma}
Note that ${\hat {\bm \beta}^{\text{RL}}}$ is an unbiased estimator for $\bm \beta$, and hence
\begin{equation}
    \begin{aligned}
    \mathbb{E}\|{{\hat {\bm \beta}^{\text{RL}}} - \bm \beta}^{}\|^2 & = \mathbb{E}\left[\sum_{i=1}^d (\hat\beta_i^{\text{RL}} - \beta_i)^2\right] \\
    & = \sum_{i=1}^d \mathbb{E} (\hat\beta_i^{\text{RL}} - \beta_i)^2 \\
    & = \sum_{i=1}^d \operatorname{Var}(\hat\beta_i^{\text{RL}})\\
    & = \operatorname{tr}(\Sigma^{\text{RL}}).
    \end{aligned}
\end{equation}
\end{proof}

\begin{proof}[Proof of Theorem~\ref{thm:ngd-bound}]
To establish an upper bound of the excess error of the private estimator, i.e., $\|{\bm{\hat\beta}^{\text{priv}}} - {\hat {\bm \beta}^{\text{RL}}}\|^2$, for Algorithm~\ref{alg:ngd}, we work under the event $\mathcal{E} = \{\Pi_{R}(z_i) = z_i, \forall i \in [n]\}$.
By the concentration bound of the Gaussian distribution, with the choice of $R = \sigma \sqrt{2\ln n}$, $\mathbb{P}(\mathcal{E}) \geq 1 - c_1\exp(-c_2 \ln n)$ where $c_1$ and $c_2$ are constants.

Recall that the loss function
$\mathcal{L}_n(\bm \beta) \stackrel{def}{=} \frac{1}{2n}(\bm z - W\bm\beta)^\top (\bm z - W\bm\beta)$.
The assumption about the eigenvalues of $W^\top  W/n$ implies that $\mathcal{L}_n(\bm \beta)$ is $\frac{L}{d}$-smooth and $\frac{1}{dL}$-strongly convex. 
See the proof in Lemma~\ref{smoothconvex}. Under $\mathcal{E}$, the iterate $\bm{\beta}^{t+1} = \Pi_C(\bm{\beta}^{t} - \eta\nabla \mathcal{L}_n(\bm{\beta}^{t}) + \bm u_t)$.
Let
$\hat {\bm\beta}^{t+1} = \Pi_C(\bm{\beta}^{t} - \eta\nabla \mathcal{L}_n(\bm{\beta}^{t}))$ be the unperturbed iterate, then $\|\bm{\beta}^{t+1} - \hat {\bm\beta}^{t+1}\| \leq \|\bm u_t\|$.
Since (A3) says that $\|\bm\beta\| \leq c_0$,
we can assume $\|{\hat {\bm \beta}^{\text{RL}}}\| \leq c_0$ without loss of generality where ${\hat {\bm \beta}^{\text{RL}}} = \underset{\bm \beta}{\arg\min}\mathcal{L}_n(\bm \beta)$.
Then, ${\hat {\bm \beta}^{\text{RL}}}$ is the same as ${\hat {\bm \beta}^{\text{RL}}} \stackrel{def}{=}\underset{\|\bm \beta\| \leq C}{\arg\min}\mathcal{L}_n(\bm \beta)$ by setting $C = c_0$.
By Lemma~\ref{conv}, with $\eta = \frac{d}{L}$, it then follows that
$$
\|\hat{\bm \beta}^{t+1} - {\hat {\bm \beta}^{\text{RL}}}\|^2 \leq \left(1- \frac{1}{L^2}\right) \|{\bm \beta}^{t} - {\hat {\bm \beta}^{\text{RL}}}\|^2.
$$
Let $c \geq 2$ be some constant, then by Lemma~\ref{lem:tri-ineq}, we have the following for the noisy iterate ${\bm \beta}^{t+1}$:
\begin{equation*}
    \begin{aligned}
    \|\bm{\beta}^{t+1} -{\hat {\bm \beta}^{\text{RL}}} \|^2 
    & = \| \hat {\bm \beta}^{t+1} - {\hat {\bm \beta}^{\text{RL}}} 
    +\bm{\beta}^{t+1} - \hat {\bm \beta}^{t+1} \|^2 \\
    & \leq \left(1 + \frac{1}{cL^2}\right)\|\hat {\bm \beta}^{t+1} -{\hat {\bm \beta}^{\text{RL}}} \|^2 + \left(1 + cL^2\right)\|\bm u_t \|^2 \\ 
    &\leq \left(1 + \frac{1}{cL^2}\right)\left(1 - \frac{1}{L^2}\right)\| {\bm \beta}^{t} -{\hat {\bm \beta}^{\text{RL}}} \|^2 +  \left(1 + cL^2\right) \|\bm u_t \|^2 \\
    &\leq \left(1 - \frac{c-1}{cL^2}\right)\| {\bm \beta}^{t} -{\hat {\bm \beta}^{\text{RL}}} \|^2 +  \left(1 + cL^2\right) \|\bm u_t \|^2.
    \end{aligned}
\end{equation*}

The above recursive formula yields
\begin{equation*}
    \|{\bm{\hat\beta}^{\text{priv}}} -{\hat {\bm \beta}^{\text{RL}}} \|^2 \leq \left(1 - \frac{c-1}{cL^2}\right)^T \|{\bm \beta}^{0} -{\hat {\bm \beta}^{\text{RL}}} \|^2
    +  \left(1 + cL^2\right) \sum_{k = 0}^{T-1} \left(1 - \frac{c-1}{cL^2}\right)^k \|\bm u_{T-1-k} \|^2.
\end{equation*}

Setting $\displaystyle T = \left\lceil \ln(c_0^2 n) / \ln \left(1 - \frac{c-1}{cL^2}\right) \right\rceil$, the first term $\displaystyle\left(1 - \frac{c-1}{cL^2}\right)^T \|{\bm \beta}^{0} -{\hat {\bm \beta}^{\text{RL}}} \|^2 \leq \frac{1}{n}$, given that $\bm\beta^0 = \bm 0$ and $\|{\hat {\bm \beta}^{\text{RL}}}\| \leq c_0$. 
When $c$ is sufficiently large, $T$ can be set to $\left\lceil L^2\ln(c_0^2 n)  \right\rceil$.


To control the second term, we apply Lemma~\ref{sum_of_chi} with $\rho = 1-(c-1)/(cL^2)$ and $\zeta = 2{\eta}^2TB^2 \frac{\ln(1/\delta)}{n^2\epsilon^2}$ which is the variance of $\bm u_{T-1-k}$ for $k = 1,..., T-1$. Provided $\delta = o(1/n)$, let $s = K_1\zeta d$ for some sufficiently large constant $K_1$ so $\frac{\rho\zeta d}{1-\rho} +  s  = \Theta(\zeta d)$ where \begin{equation*}
    \begin{aligned}
    \zeta d & = 
    4{\eta}^2TB^2 \frac{\ln(1/\delta)}{n^2\epsilon^2} \cdot d \\
    & =
    4\left(\frac{d}{L}\right)^2 \left\lceil L^2 \ln(c_0^2n) \right\rceil (\sigma \sqrt{2\ln n} c_{\bm x}  (M+4) + 2c_0c^2_{\bm x}(M+2))^2 \frac{\ln(1/\delta)}{n^2\epsilon^2} \cdot d \\
    &= O\left( \frac{\sigma^2d^3 \ln^2 n \ln(1/\delta)}{n^2\epsilon^2}\right).     
    \end{aligned}
\end{equation*}

That is, the noise term $\displaystyle\sum_{k = 0}^{T-1} \left(1 - \frac{c-1}{cL^2}\right)^k \|\bm u_{T-1-k} \|^2$ is then controlled by a corresponding big-O statement with probability at least $1-e^{-c_3 d}$, 
where $c_3 = \min \left( \frac{(1 - \rho^2)K_1^2}{8\rho^2}, \frac{K_1}{8\rho}\right)$, hence 
\begin{equation*}
    \|{\bm{\hat\beta}^{\text{priv}}} -{\hat {\bm \beta}^{\text{RL}}} \|^2 = \frac{1}{n} + O\left( \frac{\sigma^2d^3 \ln^2 n \ln(1/\delta)}{n^2\epsilon^2}\right).
\end{equation*}

\end{proof}

\begin{proof}[Proof of Theorem~\ref{thm:ssp-bound}]

Algorithm~\ref{alg:ssp} is also analyzed under    
$\mathcal{E} = \{\Pi_R(z_i) = z_i, \forall i = 1, \dots, n \}$ as in Theorem~\ref{thm:ssp-bound}. Then, with $R = \sigma \sqrt{2\ln n}$, we have
$\mathbb{P}(\mathcal{E}) \geq 1 - c_1\exp(-c_2 \ln n)$.

By Lemma~\ref{lem: matrix}, 
\begin{equation}
\begin{aligned}
    \quad &  \quad (W^\top W + U )^{-1} 
    & = ( W^\top W )^{-1} - ( W^\top W )^{-1} \cdot (I + U(W^\top W)^{-1})^{-1} \cdot U(W^\top W)^{-1}.
\label{eq:matrixinverse}
\end{aligned}
\end{equation}
Then,
\begin{equation}
\begin{aligned}
    {\bm{\hat\beta}^{\text{priv}}} - {\hat {\bm \beta}^{\text{RL}}} & = (W^\top W + U )^{-1}(W^\top \bm z + \bm u) - (W^\top W )^{-1}W^\top \bm z\\
    & =  (W^\top W )^{-1} \bm u - ( W^\top W )^{-1} \cdot (I + U(W^\top W)^{-1})^{-1} \cdot U(W^\top W)^{-1}(W^\top \bm z + \bm u) \\
    & = (W^\top W )^{-1} \bm u - ( W^\top W )^{-1} \cdot (I + U(W^\top W)^{-1})^{-1} \cdot U(\bm {\hat {\bm \beta}^{\text{RL}}} + (W^\top W)^{-1} \bm u) 
    \label{eq:privest-rlest}
\end{aligned}
\end{equation}

To bound $ \|{\bm{\hat\beta}^{\text{priv}}} - {\hat {\bm \beta}^{\text{RL}}}\|$, we need to bound the norms of $(W^\top W)^{-1}$,  $U$, $\bm u$, and $(I + U(W^\top W)^{-1})^{-1}$.
First, by the assumption (A4): \begin{equation*}
0< \frac{1}{L} < d\lambda_{\min}\left(\frac{W^\top  W}{n} \right) \leq d\lambda_{\max}\left(\frac{W^\top  W}{n} \right)  < L
\end{equation*}
for some constant $1 < L < \infty$, we have
 $$ \frac{d}{Ln} \leq \lambda_{\min}( W^\top  W)^{-1} \leq \lambda_{\max}( W^\top  W)^{-1} \leq \frac{Ld}{n}.$$
Then,
\begin{equation}
    \| (W^\top W)^{-1}\| = \lambda_{\max}( W^\top  W)^{-1} = O\left(\frac{d}{n}\right).
    \label{norm:W}
\end{equation}
Recall from Algorithm \ref{alg:ssp} that $U$ is a Gaussian symmetric matrix with upper triangle given by iid $N(0,\omega).$ The Gaussian concentration bounds give that
w.p. $\geq 1 - e^{-c_3d}$, we have
\begin{equation}
    \| U \| = O(\omega \sqrt{d}) = O\left(\frac{\sigma^2 \ln n\sqrt{d\ln(1/\delta)}}{ \epsilon} \right).
    \label{norm:U}
\end{equation}
The result (\ref{norm:U}) is from random matrix theory 
\cite[Corollary 4.4.8]{Vershynin}.
Since the vector norm is bounded by the norm of the matrix that contains the vector as a column, thus we also have 
\begin{equation}
    \| \bm u \| 
    = O\left(\frac{\sigma^2 \ln n\sqrt{d\ln(1/\delta)}}{ \epsilon} \right).
    \label{norm:n}
\end{equation}

For $\|(I + U(W^\top W)^{-1})^{-1}\|$, 
consider its Taylor series:
\begin{equation}
     (I + U(W^\top W)^{-1})^{-1} = \sum_{i = 0}^\infty (- U(W^\top W)^{-1})^i
\label{eq:tsofinverse}
\end{equation}
W. p. $\geq 1-e^{-c_3d}$, $$\|U(W^\top W)^{-1}\| \leq \|U\|\|(W^\top W)^{-1}\| = O\left(\frac{\sigma^2 d \ln n \sqrt{d\ln(1/\delta)}}{n \epsilon} \right)$$  is going to zero as $n \to \infty$. Therefore, the Taylor Series (\ref{eq:tsofinverse}) converges.
Then,
\begin{equation}
    \begin{aligned}
         \|(I + U(W^\top W)^{-1})^{-1}\| & = \| \sum_{i = 0}^\infty (- U(W^\top W)^{-1})^i \| \\
         & \leq \sum_{i = 0}^\infty \| (- U(W^\top W)^{-1})^i \| \\
         & \leq \sum_{i = 0}^\infty \| U(W^\top W)^{-1}) \| ^i \\
         & = \frac{1}{1-\| U(W^\top W)^{-1})^{-1}\|}
         & = O(1)
    \end{aligned}
    \label{norm:the-inverse}
\end{equation}
 w.p. $\geq 1-e^{-c_3d}$. 
Plugging all the bounds into (\ref{eq:privest-rlest}),
we have derived
\begin{equation*}
    \|{\bm{\hat\beta}^{\text{priv}}} - {\hat {\bm \beta}^{\text{RL}}}\|^2 = O\left(\frac{\sigma^4 d^3 \ln^2 n  \ln(1/\delta)}{n^2 \epsilon^2} \right).
\end{equation*}
 
 \end{proof}

\begin{proof}[Proof of Corollary~\ref{cor:final.bound}]

The proof is completed by using the inequality
\begin{equation*}
\begin{aligned}
    \|{ {\bm{\hat\beta}^{\text{priv}}} - \bm \beta}^{}\|^2 & \leq 2(\|{{\hat {\bm \beta}^{\text{RL}}} - \bm \beta}^{}\|^2 + \|{ {\bm{\hat\beta}^{\text{priv}}} - {\hat {\bm \beta}^{\text{RL}}}}^{}\|^2 ).
    \end{aligned}
\end{equation*}
\end{proof}

\begin{proof}[Proof of Theorem~\ref{thm:var-ngd}]
Consider the non-projected estimator
\begin{equation}
     \bm{\beta}^{t+1} = \bm{\beta}^t - 
    \frac{\eta}{n}W^\top  (W \bm{\beta}^t - \bm z) + \bm u_t.
\end{equation}
Let $\displaystyle A \stackrel{def}{=} \frac{\eta}{n}W^\top  W$ and $\displaystyle B \stackrel{def}{=} \frac{\eta}{n}W$. From the recursive form $\bm{\beta}^{t+1} = (I_d - A) \bm{\beta}^{t} + B^\top \bm z + \bm u_t$, we derive for $ {\bm{\hat\beta}^{\text{priv}}} \stackrel{def}{=} \bm{\beta}^{T}$:
\begin{equation}
    {\bm{\hat\beta}^{\text{priv}}} = (I_d - A)^{T}\bm{\beta}^0 + \sum_{t=1}^T (I_d - A) ^{t-1}( B^\top\bm z + \bm u_{T-t}).
\end{equation}
Then,
\begin{equation}
    \operatorname{Var}({\bm{\hat\beta}^{\text{priv}}}) = \sum_{t=1}^T (I_d - A) ^{t-1} \cdot B^\top \Sigma_{\bm z} B  \cdot \sum_{t=1}^T (I_d - A) ^{t-1}+ \omega^2 \sum_{t=1}^T (I_d - A) ^{2t-2}.
\end{equation}
\end{proof}

\begin{proof}[Proof of Theorem~\ref{thm:var-ssp}]

By rewriting ${\bm{\hat\beta}^{\text{priv}}}$ as in (\ref{eq:privest-rlest}) and ignoring the remainder of the first-order Taylor expansion in (\ref{eq:tsofinverse}), 
we have
\begin{equation}
\begin{aligned}
    {\bm{\hat\beta}^{\text{priv}}} & =  {\hat {\bm \beta}^{\text{RL}}} + (W^\top W )^{-1} \bm u - (W^\top W )^{-1} U \bm {\hat {\bm \beta}^{\text{RL}}} - ( W^\top W )^{-1} U (W^\top W)^{-1} \bm u.
\end{aligned}
\end{equation}
Then,
\begin{equation}
\begin{aligned}
    \operatorname{Var}({\bm{\hat\beta}^{\text{priv}}}) & =  \operatorname{Var}({\hat {\bm \beta}^{\text{RL}}}) + \operatorname{Var}[(W^\top W )^{-1} \bm u] + \operatorname{Var}[(W^\top W )^{-1} U \bm {\hat {\bm \beta}^{\text{RL}}}]\\ & \quad +   \operatorname{Var}[(W^\top W )^{-1}U (W^\top W)^{-1} \bm u]  - 2 \operatorname{Cov}({\hat {\bm \beta}^{\text{RL}}}, (W^\top W )^{-1}U \bm {\hat {\bm \beta}^{\text{RL}}}) \\
    & \quad - 2 \operatorname{Cov}((W^\top W )^{-1} \bm u, ( W^\top W )^{-1} U (W^\top W)^{-1} \bm u)\\
    & \quad + 2 \operatorname{Cov}((W^\top W )^{-1} U \bm {\hat {\bm \beta}^{\text{RL}}}, ( W^\top W )^{-1} U (W^\top W)^{-1} \bm u).
    \label{eq:variance-prive}
\end{aligned}
\end{equation}
Since $\bm u \sim \mathcal{N}(0, \omega^2 I_d)$, we have
\begin{equation*}
    \operatorname{Var}((W^\top W )^{-1} \bm u)  = \omega^2 (W^\top W )^{-2}.
\end{equation*}
For the third term, 
let $\Sigma_1 \stackrel{def}{=}  \operatorname{Var}(U \bm {\hat {\bm \beta}^{\text{RL}}})$. 
By Lemma~\ref{lem:xy-var},
\begin{equation*}
    (\Sigma_1)_{kk} = \omega^2 \sum_{i=1}^d(\beta_i^2 + \Sigma^{\text{RL}}_{ii})
\end{equation*}
and 
\begin{equation*}
    (\Sigma_1)_{kl} = \omega^2 (\beta_k\beta_l + \Sigma^{\text{RL}}_{kl}).
\end{equation*}
For the fourth term, let $\Sigma_2 \stackrel{def}{=}  \operatorname{Var}(U (W^\top W)^{-1} \bm u)$ and let $\Sigma' = \operatorname{Var}((W^\top W)^{-1} \bm u) = \omega^2(W^\top W)^{-2}$. 
\begin{equation*}
    (\Sigma_2)_{kk} = \omega^2 \sum_{i=1}^d \Sigma'_{ii}
\end{equation*}
and 
\begin{equation*}
    (\Sigma_2)_{kl} = \omega^2 \Sigma'_{kl}.
\end{equation*}

On the other hand, the covariances in (\ref{eq:variance-prive}) are all zeros by Lemmas \ref{lem:cov1} and \ref{lem:cov2} due to the independencies and zero expectations of $U$ and $\bm u$.
Then putting them together, we have
\begin{equation}
    \Sigma = \Sigma^{\text{RL}}  + (W^\top W)^{-1} (\omega ^2I_d + \Sigma_1 + \Sigma_2) (W^\top W)^{-1}.
\end{equation}
Note that $\Sigma_1$ and $\Sigma_2$ have a common factor $\omega^2$. By rescaling $\Sigma_1$ and $\Sigma_2$, we have the expression for $\Sigma$ as stated in the theorem. 

\end{proof}

\printbibliography

\end{document}